\documentclass[transmag]{IEEEtran}

\usepackage{stmaryrd}
\usepackage{amsfonts}
\usepackage{mathrsfs,amsmath}
\usepackage{amssymb}
\usepackage{makeidx}
\usepackage{multirow}
\usepackage{algorithmic} 
\usepackage{bm}
\usepackage{setspace}
\usepackage{amsthm}
\usepackage{booktabs}
\usepackage{authblk}
\usepackage[ruled,linesnumbered,vlined]{algorithm2e}  
\usepackage{color}
\usepackage{enumerate}

\usepackage{tabularx,booktabs}
 
\usepackage{ifpdf}
 \ifpdf
 \else
 \fi
 
\usepackage{cite}
 
\ifCLASSINFOpdf
   \usepackage[pdftex]{graphicx}
\else	
\fi
 
\usepackage{amsmath}
\usepackage[percent]{overpic}

\usepackage{array} 
\usepackage{tikz,pgfplots,filecontents}
\usepackage{pgfplots}
\usetikzlibrary{spy} %
\pgfplotsset{width=9cm, height=7.2cm, compat=1.9}
 
\ifCLASSOPTIONcompsoc 
  \usepackage[caption=false,font=normalsize,labelfont=sf,textfont=sf]{subfig}
\else
  \usepackage[caption=false,font=footnotesize]{subfig}
\fi
 
\usepackage{url}
 
\makeatletter
\let\NAT@parse\undefined
\makeatother
\usepackage{hyperref}  

\tikzset{elegant/.style={smooth,thick,samples=50,cyan}}
\tikzset{eaxis/.style={->,>=stealth}}

\newtheorem{theorem}{Theorem}
\newtheorem{lemma}[theorem]{Lemma}
\newtheorem{proposition}[theorem]{Proposition}

 \newtheorem{remark}[theorem]{Remark}

\IEEEoverridecommandlockouts   

\allowdisplaybreaks

\pgfplotsset{every axis legend/.style={%
cells={anchor=west},
inner xsep=3pt,inner ysep=2pt,nodes={inner sep=0.8pt,text depth=0.15em},
anchor=north east,%
shape=rectangle,%
fill=white,%
draw=black,
at={(0.98,0.98)},
font=\footnotesize,
}}

\pgfplotsset{every axis/.append style={line width=0.6pt,tick style={line width=0.8pt}}}

\usepackage{latexsym}
\usepackage{graphicx}
\usepackage{amsfonts,amssymb,amsmath}
\usepackage{hyperref}
\def\BibTeX{{\rm B\kern-.05em{\sc i\kern-.025em b}\kern-.08em T\kern-.1667em\lower.7ex\hbox{E}\kern-.125emX}}
\begin{document}

\title{Resource Optimization with Interference Coupling in Multi-RIS-assisted Multi-cell Systems}

\author{Zhanwei Yu (\IEEEmembership{Student Member, IEEE}) and Di Yuan (\IEEEmembership{Senior Member, IEEE})
\thanks{This work has been supported by the Swedish Research Council under Grant 2018-05247.}
\thanks{Z. Yu and D. Yuan are with Department of Information Technology, Uppsala University, 751 05 Uppsala, Sweden (e-mail: \{zhanwei.yu; di.yuan\}@it.uu.se).}
}

\IEEEtitleabstractindextext{\begin{abstract}Deploying reconfigurable intelligent surface (RIS) to enhance wireless transmission is a promising approach. In this paper, we investigate large-scale multi-RIS-assisted multi-cell systems, where multiple RISs are deployed in each cell. Different from the full-buffer scenario, the mutual interference in our system is not known a priori, and for this reason we apply the load coupling model to analyze this system. The objective is to minimize the total resource consumption subject to user demand requirement by optimizing the reflection coefficients in the cells. The cells are highly coupled and the overall problem is non-convex. To tackle this, we first investigate the single-cell case with given interference, and propose a low-complexity algorithm based on the Majorization-Minimization method to obtain a locally optimal solution. Then, we embed this algorithm into an algorithmic framework for the overall multi-cell problem, and prove its feasibility and convergence to a solution that is at least locally optimal. Simulation results demonstrate the benefit of RIS in time-frequency resource utilization in the multi-cell system.
\end{abstract}

}

\maketitle

\section{Introduction} \label{Sec:Intro}

\IEEEPARstart{A}{n} array of new technologies that can contribute to beyond 5G (B5G) and the sixth-generation (6G) wireless networks are being proposed and extensively investigated, including massive multiple-input multiple-output (m-MIMO) \cite{BJORNSON20193} or utilizing higher frequency bands such as millimeter wave (mmWave) \cite{mmWave} and even terahertz (THz) \cite{Terahertz}. However, to apply these technologies, it may be necessary to develop new network software and hardware platforms. For example, millions of antennas and baseband units (BBU) need to be deployed to support mmWave and THz technologies. As a low-cost solution, reconfigurable intelligent surface (RIS) has recently drawn significant attention \cite{IRS1,IRS2, IRS3, IRS-RIS}.

The RIS-assisted wireless network not only provides lower cost compared with some other technologies, but also leads to less overhead to the existing wireless systems \cite{wu2020tutorial}. RIS is a planar surface composed of reconfigurable passive printed dipoles connected to a controller. The controller is able to reconfigure the phase shifts according to the incident signal, to achieve more favorable signal propagation. By its nature, RIS provides a higher degree of freedom for data transmission, thus improving the capacity and reliability. In contrast to traditional relay technology, RIS consumes much less energy without amplifying noise or generating self-interference, as the signal is passively reflected. Also, as deploying RIS in a wireless network is modular, it is more suitable for upgrading current wireless systems.

\subsection{Related Works}

\subsubsection{Studies on Characteristics of RIS} The authors of \cite{Emil1} propose a far-field path loss model for RIS by the physical optics method, and explain why RIS can beamform the diffuse signal to the desired receivers. In \cite{9119122, bjornson2019intelligent, huang2019reconfigurable}, the differences and similarities among RIS, decode-and-forward (DF), and amplify-and-forward (AF) relay are discussed. The authors of \cite{bjornson2019intelligent} point out that a sufficient number of reconfigurable elements should be considered to make up for the low reflecting channel gain without any amplification. The authors of \cite{8970580} conduct the performance comparison between non-orthogonal multiple access (NOMA) and orthogonal multiple access (OMA) in the RIS-assisted downlink for the two-users scenario, and reveal that OMA outperforms NOMA for near-RIS user pairing. Additionally, even though the adjustment range of the reflection elements usually is limited in order to reduce the cost in practice, \cite{8746155} shows that it still can guarantee a good spectral efficiency. Furthermore, in \cite{zou2020joint}, RIS powered by wireless energy transmission is considered, and the results show that RIS can work well with wireless charging.

\subsubsection{Studies on Single-cell System with RIS} Inspired by the advantages of RIS, the authors of \cite{wu2019towards} outline several typical use cases of RIS, {e.g.}, improving the experience of the user located in poorly covered areas and enhancing the capacity and reliability of massive devices communications. The authors of \cite{9039554} optimize the reflection elements to improve the rate of the cell-edge users in an RIS-assisted orthogonal frequency division multiplexing access (OFDMA)-based system. In \cite{wu1}, the authors jointly adopt the semidefinite relaxation (SDR) and the alternating optimization (AO) method to improve the spectral and energy efficiency of an RIS-assisted multi-user multi-input single-output (MISO) downlink system. For the same type of system, the authors of \cite{guo2019weighted} propose an algorithm based on the fractional programming technique to improve the capacity. Besides, due to its excellent compatibility, the RIS is able to adapt to various multiple access techniques, e.g., NOMA \cite{9203956, 9240028} and the promising rate splitting multiple access (RSMA) \cite{yang2020energy}. Additionally, a variety of RIS-assisted application are investigate in \cite{9076830, 9110849, 9133107, 8743496, 9133130, Huang_additional}. The works in \cite{CR1, CR2, CR3, CR4, CR5} study RIS-assisted spectrum sharing scenarios. The authors of \cite{CR1} integrate an RIS into a multi-user full-duplex cognitive radio network (CRN) to simultaneously improve the system performance of the secondary network and efficiently reduce the interference from the primary users (PUs). The authors of \cite{CR2} provide an algorithm based on block coordinate descent to maximize the weighted sum rate of the secondary users (SUs) in CRNs subject to their total power. The work in \cite{CR3} investigates two types of CSI error models for the PU-related channels in RIS-assisted CRNs, and propose two schemes based on the successive convex approximation method to jointly optimize the transmit precoding and phase shift of matrices. The authors of \cite{CR4} propose alternative optimization method based semidefinite relaxation techniques via jointly optimizing vertical beamforming and RIS to maximize spectral efficiency of the secondary network in an RIS-assisted CRN. The work in \cite{CR5} studies an RIS-assisted spectrum sharing underlay cognitive radio wiretap channel, and proposes efficient algorithms to enhance the secrecy rate of SU for three different cases: full CSI, imperfect CSI, and no CSI. In addition, the authors of \cite{learning1} develop a deep reinforcement learning (DRL) based algorithm to obtain the transmit beamforming and phase shifts in RIS-assisted multi-user MISO systems. A novel DRL-based hybrid beamforming algorithm is designed in \cite{learning2} to improve the coverage range of THz-band frequencies in multi-hop RIS-assisted communication systems. The authors of \cite{learning3} propose a decaying deep Q-network based algorithm to solve an energy consumption minimizing problem in RIS in unmanned aerial vehicle (UAV) enabled wireless networks.

\subsubsection{Studies on Multi-cell System with RIS} All the above works \cite{wu2019towards, 9039554, wu1, guo2019weighted, 9203956, 9240028, yang2020energy, 9076830,  9110849, 9133107, 8743496, 9133130, CR1, CR2, CR3, CR4, CR5} focus on a single-cell setup. Different from the single-cell scenario, the key issue for a multi-cell system is the presence of inter-cell interference. The authors of \cite{hua2020intelligent} jointly optimize transmission beamforming and the reflection elements to maximize the minimum user rate in an RIS-assisted joint processing coordinated multipoint (JP-CoMP) system. For fairness, a max-min weighted signal-interference-plus-noise ratio (SINR) problem in an RIS-assisted multi-cell MISO system is solved by three AO-based algorithms in \cite{xie2020max}, and the numerical results demonstrate that RIS can help improve the SINR and suppress the inter-cell interference, especially for cell-edge users. Three algorithms are proposed to minimize the sum power of a large-scale discrete-phase RIS-assisted MIMO multi-cell system in \cite{omid2020irs}. In \cite{ni2020resource}, the authors give a novel algorithm for resource allocation in an RIS-assisted multi-cell NOMA system to maximize the sum rate. The authors of \cite{kim2020exploiting} jointly optimize the reflection elements, base station (BS)-user pairing, and user transmit power by a DRL approach to maximize the sum rate for a multi-RIS-assisted MIMO multi-cell uplink system. 
 
\subsection{Our Contributions}

It is noteworthy that almost all existing papers studying RIS-assisted multi-cell systems focus on the rate maximization problem, which usually means making full use of the resource. In many scenarios, however, the user demands are finite, and hence it is relevant to meet the user data demands rather than achieving highest total rate by exhausting all the resource. For such scenarios, inter-cell interference is not represented by the worst-case value as not all resources are used for transmission. Having this in mind, we consider the time-frequency resource consumption minimization problem in the multi-RIS-assisted multi-cell system. Different from the existing full-buffer works, in our scenario, the mutual interference is not known a priori instead of the worst-case value. The main contributions of this paper can be summarized as follows:

\begin{itemize}
\item We formulate an optimization problem for the multi-RIS-assisted multi-cell system. Our objective is to optimize the reflection coefficients of all RISs in the system to minimize the total time-frequency resource consumption subject to the user demand requirement.
\item Due to the non-convexity of this problem, we first investigate its single-cell version. We derive an approximate convex model for the single-cell problem. We then propose an algorithm based on the Majorization-Minimization (MM) method to obtain a locally optimal solution.
\item In the next step, we embed the single-cell algorithm into an algorithmic framework to obtain a locally optimal solution for the overall multi-cell problem, and prove its feasibility and convergence. The algorithmic framework also is proved that it can reach the global optimality if the single-cell can be solved to optimum.
\item We evaluate the performance of the system optimized by our algorithmic framework, make performance comparison to three benchmark solutions. The numerical results demonstrate the proposed algorithmic framework is capable of achieving significant time-frequency resource saving.
\end{itemize}

\subsection{Organization and Notation}

\subsubsection{Organization} The remainder of this paper is organized as follows. In Section \ref{Sec:SystemModel}, we describe the multi-RIS-assisted multi-cell system model and formulate our optimization problem with the load coupling model for characterizing inter-cell interference. In Section \ref{Sec:SingleCell}, we investigate the single-cell problem, and propose an algorithm based on the MM method. In Section \ref{Sec:multi-cell}, we propose an algorithmic framework for the multi-cell problem. In Section \ref{Sec:evaluation}, simulation results are shown for performance evaluation. Finally, we conclude in Section \ref{Sec:conclusion}.

\subsubsection{Notation} The matrices and vectors are respectively denoted by boldface capital and lower case letters. $\mathbb{C}^{1 \times M}$ and $\mathbb{C}^{M \times 1}$ stand for a collection of complex matrices, in which each matrix has size $1 \times M$ and $M \times 1$, respectively. $\mathfrak{diag} \{{\cdot}\}$ denotes the diagonalization operation. For a complex value $\mathrm{e}^{\mathrm{i}\theta}$, $\mathrm{i}$ denotes the imaginary unit. $x\sim \mathcal{CN}(\mu,\sigma^2)$ denotes the circularly symmetric complex Gaussian (CSCG) distribution with mean $\mu$ and variance $\sigma^2$. Transpose, conjugate, and transpose-conjugate operations are denoted by $(\cdot)^T$, $(\cdot)^\star$, and $(\cdot)^H$, respectively. $\Re\{\cdot\}$ and $\Im\{\cdot\}$ denote the real and imaginary parts of a complex number, respectively. $\mathcal{O}(\cdot)$ denotes the order in computational complexity. In addition, $||\cdot||_{\infty}$ denotes the infinity norm of a vector.


\section{System Model and Problem Formulation} \label{Sec:SystemModel}

\subsection{System Model}\label{Preliminaries}

\begin{figure}[tbp]
\centering
\begin{overpic}
[scale=0.16]{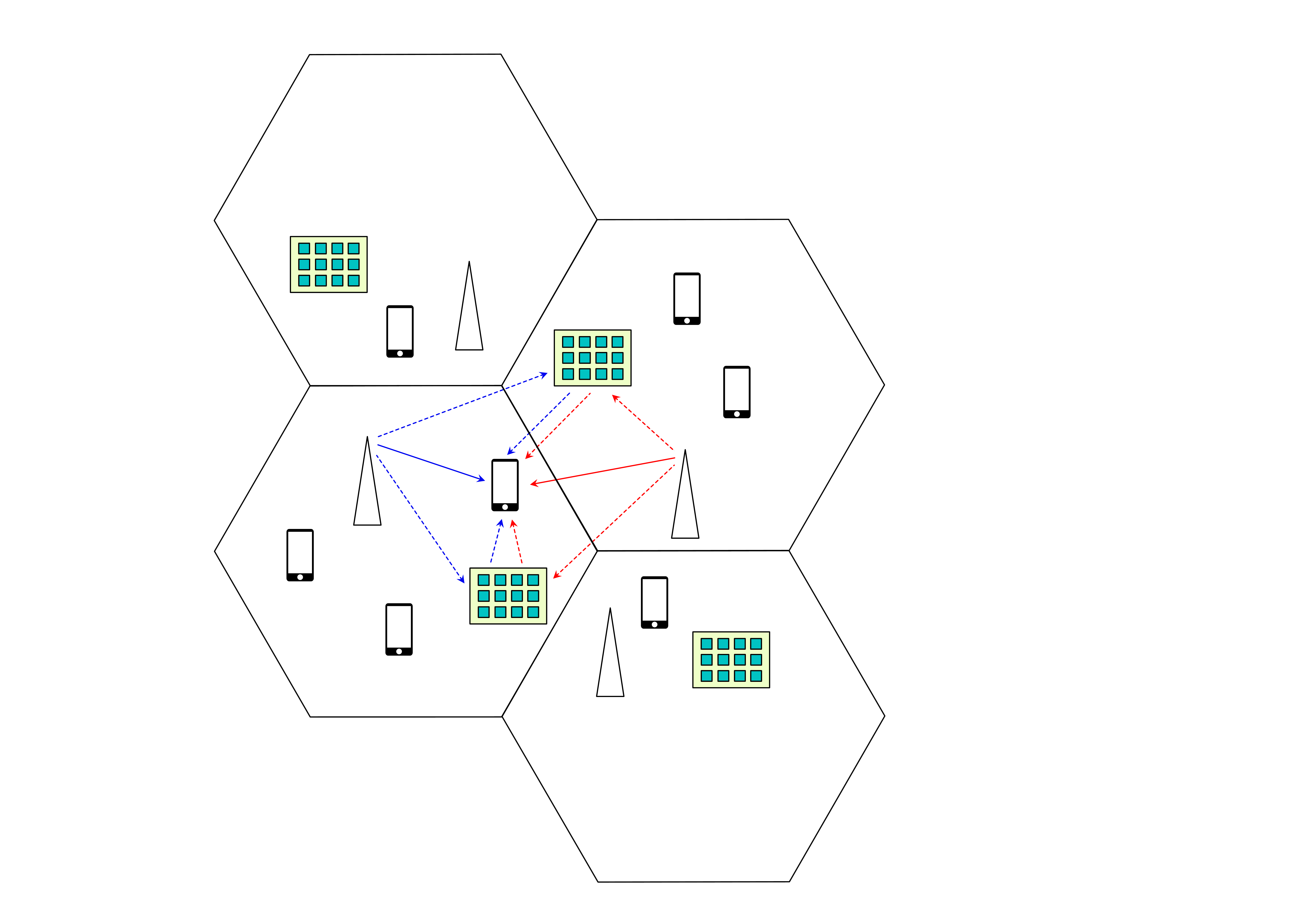}
\put(38.5,46){$g_{i,j}$}
\put(35.5,36.5){$\boldsymbol{H}_{l,i}$}
\put(42,30){$\boldsymbol{G}_{l,j}$}
\end{overpic}
\caption{This figure shows a concise example for the received signal and interference at a user in the multi-RIS-assisted multi-cell system, where blue solid line, dashed line stand for the direct and reflected signal, respectively, and red solid line, dashed line stand for the direct and reflected interference, respectively.}\label{fig:sys}
\end{figure}

We consider a downlink multi-cell wireless system with a total of $L$ RISs and $I$ user equipments (UEs) distributed in $J$ cells. Denote by $\mathcal{L} = \{1,2,...,L \}$, $\mathcal{I} = \{1,2,...,I \}$, and $\mathcal{J} = \{1,2,...,J\}$ the sets of RISs, cells, and UEs, respectively. In addition, $\mathcal{J}_i$ $(\forall i \in \mathcal{I})$ represents the set of the UEs served by cell $i$. Without loss of generality, we assume that all RISs have the same number (denoted by $M$) of reflection elements. Let $\mathcal{M} = \{1,2,...,M\}$. 

We use $g_{ij}$, $\boldsymbol{G}_{il} \in \mathbb{C}^{1 \times M}$, and $\boldsymbol{H}_{lj}\in \mathbb{C}^{M \times 1}$ to denote the channel gain from the BS of cell $i$ to UE $j$, that from the BS of cell $i$ to RIS $l$, and that from RIS $l$ to UE $j$, respectively. Note that the channels between BS and RIS, as well as the channels between RIS and users in RIS-assisted multi-cell systems can be estimated by existing works, such as the channel estimation framework based on the PARAllel FACtor (PARAFAC) decomposition in \cite{Estimation}. The diagonal reflection matrix of the $l$-th RIS is denoted by 
\begin{equation}
\boldsymbol{\Theta}_l =\mathfrak{diag} \{\phi_{l1}, \phi_{l2} , ..., \phi_{lM} \},
\end{equation}
where $\phi_{lm} = \lambda_{lm} \mathrm{e}^{\mathrm{i}\theta_{lm}}$, $\forall m \in \mathcal{M}$. Further, $\phi_{lm}$ is the $m$-th reflection coefficient of the $l$-th RIS, where $\lambda_{lm}$ and $\theta_{lm}$ represent its amplitude and phase, respectively. In this paper, we consider the following three value domains of the reflection coefficient. 
 \subsubsection{Ideal} Ideally, the amplitude and phase can be adjusted independently, i.e., $\lambda_{lm} \in \left[0,1\right]$ and $\theta_{lm} \in [0, 2\pi]$. This leads to the following domain definition:
\begin{equation}
\mathcal{D}_1=\left\{ \lambda_{lm}  \mathrm{e}^{\mathrm{i}\theta_{lm}}\Big| \lambda_{lm} \in [0,1], \theta_{lm} \in [0,2\pi]\right\}.
\end{equation}
\subsubsection{Continuous Phase Shifter} In this case, the amplitude is at its maximum, {i.e.}, $\lambda_{lm} = 1$, and only the phase can be adjusted. Under this assumption, the corresponding domain is:
\begin{equation}
\mathcal{D}_2=\left\{  \mathrm{e}^{\mathrm{i}\theta_{lm}}\Big| \theta_{lm} \in [0,2\pi]\right\}.
\end{equation}
\subsubsection{Discrete Phase Shifter} This is also called practical RIS. A practical RIS only provides a finite number of phase shifts, and its amplitude is fixed at one. We assume the phase shift can be adjusted in $N$ discrete values in the following domain:
\begin{equation}
\mathcal{D}_3= \left\{  \mathrm{e}^{\mathrm{i}\theta_{lm}}\Big| \theta_{lm}\in \{0,\Delta\theta, ..., \left(N-1\right)\Delta\theta\}, N\geq2\right\},
\end{equation}
where $\Delta \theta = 2\pi / N$.

Note that, for single-user single-antenna scenarios with perfect CSI, amplitude control is unnecessary because the amplitude should be one in both theory and practice \cite{Amplitude_Control}. For the multi-user systems with perfect SCI, the work in \cite{guo2019weighted} has shown that the performance gain provided by the amplitude control is negligible. Therefore, in our paper, we do not exclude amplitude control, and use optimization to examine if our scenario with cell load coupling also submits to this observation. Our numerical results (see later in Section \ref{Sec:evaluation}) show the optimal amplitude of the RIS should be one.

We use $x_k$ to represent the transmitted signal at BS in cell $k$. Accordingly, the overall received signal including interference at UE $j$ is given by
\begin{align}
	y_j& =\underbrace{ \sum_{k \in \mathcal{I}}g_{kj}x_k}_{\text{Direct links (from BSs)}} +  \underbrace{\sum_{k \in \mathcal{I}}\sum_{l \in \mathcal{L}} \boldsymbol{G}_{kl}  \boldsymbol{\Theta}_l\boldsymbol{H}_{lj}x_k}_{\text{RIS-assisted links}} +z_j \\
	&= \sum_{k \in \mathcal{I}} \left( g_{kj} + \sum_{l \in \mathcal{L}} \boldsymbol{G}_{kl}  \boldsymbol{\Theta}_l\boldsymbol{H}_{lj}\right) x_k + z_j,\label{y1}
\end{align}
where $z_j\sim \mathcal{CN}(0,\sigma^2)$ is the additive white Gaussian noise at UE $j$. UE $j$ in cell $i$ treats all the signals from other cells as interference, then (\ref{y1}) can be re-written by
\begin{align}
	y_j& =\underbrace{ \left(g_{ij} + \sum_{l \in \mathcal{L}} \boldsymbol{G}_{il}  \boldsymbol{\Theta}_l\boldsymbol{H}_{lj}\right)x_i}_{\text{Desired signal}} \notag \\ 
	&+  \underbrace{\left( \sum_{k\in\mathcal{I}, \atop k\neq i}g_{kj} +  \sum_{k\in\mathcal{I}, \atop k\neq i}\sum_{l \in \mathcal{L}} \boldsymbol{G}_{kl}  \boldsymbol{\Theta}_l\boldsymbol{H}_{lj}\right)x_k}_{\text{Interference}} +z_j,
\end{align}
Let $P_k$ be the transmission power per resource block (RB) in cell $k$, and we use $ \rho_k \in [0,1]$ to represent the proportion of RBs consumed in cell $k$, and this entity is referred to as the cell {\em load} in \cite{mogensen2007lte}. In fact, it is quite difficult to fully coordinate the inter-cell interference in large-scale multi-cell networks. For this reason, we use the load levels to characterize a cell's likelihood of interfering the others. A cell transmitting on many RBs, i.e., high load, generates more interference to others than almost idle cell of low load. We remark that this inter-cell interference approximation is suitable for the network-level performance analysis, and \cite{fehske2012aggregation, klessig2015performance} have shown that the approximation has good accuracy for inter-cell interference characterization. Thus, the power of interference received at UE $j$ is calculated by
\begin{equation}
	\sum_{k\in\mathcal{I}, \atop k\neq i}  | g_{kj}\!+\! \sum_{l \in \mathcal{L}}\boldsymbol{G}_{kl}  \boldsymbol{\Theta}_l\boldsymbol{H}_{lj} |^2\rho_k{P_k},
\end{equation}
where $g_{kj}$ and $\sum_{l \in \mathcal{L}}\boldsymbol{G}_{kl}  \boldsymbol{\Theta}_l\boldsymbol{H}_{lj}$ are the channel gain of the direct interference link and that of the RIS-assisted interference link between cell $k$ and UE $j$, respectively.\\
Hence, the signal-to-interference-and-noise ratio (SINR) of the UE $j$ in cell $i$ is modelled as
\begin{equation}\label{sinr1}
\text{SINR}_j \left(\boldsymbol{\rho}, \boldsymbol{\phi} \right)\!=\! \frac{ \left | g_{ij} \! + \! \sum\limits_{l \in \mathcal{L}}\boldsymbol{G}_{il}  \boldsymbol{\Theta}_l\boldsymbol{H}_{lj}\right |^2{P_i} }{ \sum\limits_{k\in\mathcal{I}, \atop k\neq i} \left | g_{kj}\!+\! \sum\limits_{l \in \mathcal{L}}\boldsymbol{G}_{kl}  \boldsymbol{\Theta}_l\boldsymbol{H}_{lj}\right |^2{P_k}\rho_k \!+\! \sigma^2},
\end{equation}
where $\boldsymbol{\rho} = [\rho_1, \rho_2, ..., \rho_I ]^T$ and $\boldsymbol{\phi} = \{\phi_{lm} | l \in \mathcal{L}, m \in \mathcal{M} \}$. Note that the load of cell $k$, $\rho_k$ in (\ref{sinr1}), the proportion of resource used for transmission in cell $k$, serves as interference scaling. 

The achievable capacity of UE $j$ is then $ \log \left(1+\text{SINR}_j \left(\boldsymbol{\rho}, \boldsymbol{\phi} \right)\right)$. Let $d_j$ denote the demand of UE $j$, and denote by $B$ and $K$ the bandwidth of each RB and the total number of RBs in one cell, respectively. In addition, let $\rho_j$ be the proportion of RBs consumption by a specific UE $j$ in the associated cell. We have  
\begin{equation}
KB\rho_j \log_2\left(1+\text{SINR}_j \left(\boldsymbol{\rho}, \boldsymbol{\phi}\right)\right) \geq d_j.
\end{equation}
For any cell $i$, we have
\begin{equation}\label{rhoi}
\rho_i = \sum_{j\in \mathcal{J}_i} \rho_j \geq \sum_{j\in \mathcal{J}_i} \frac{d_j}{ K B \log_2\left(1+\text{SINR}_j \left(\boldsymbol{\rho}, \boldsymbol{\phi}\right)\right)}.
\end{equation}

For convenience, in the following discussion, we use normalized $d_j$ such that $B$ and $K$ are no longer necessary.

\subsection{Problem Formulation}\label{MathematicalFormulation}
We consider the following optimization problem where the reflection coefficients are to be optimized for minimizing the total resource consumption, subject to the user demand requirement. \begin{subequations}\label{formulation}
\begin{align}
\textup{[{\rm P1}]}\ \ \ \underset{\boldsymbol{\rho}, \boldsymbol{\phi}}{\min} \ \ & \sum_{ i \in \mathcal{I}} \rho_i   \label{P1obj} \\
\textup{ s.t.}\ \ \ & \rho_j \log_2\left(1+\text{SINR}_j \left(\boldsymbol{\rho}, \boldsymbol{\phi}\right)\right) \geq d_j, \forall j \in \mathcal{J}, \label{P1C1}\\
& \rho_i = \sum_{j\in \mathcal{J}_i} \rho_j , \forall i \in \mathcal{I},\label{P1C2}\\
& \phi_{lm} \in \mathcal{D}, \forall l \in \mathcal{L},\forall m \in \mathcal{M}.\label{P1C3} 
\end{align}
\end{subequations}

The objective function (\ref{P1obj}) is the sum of required RBs. Constraint (\ref{P1C1}) represents user demand requirement, and it is non-convex. Constraint (\ref{P1C2}) is the load of cell $i$. Constraint (\ref{P1C3}) is for the value domain of the reflection coefficients, in which $\mathcal{D}$ can be any of $\mathcal{D}_1$, $\mathcal{D}_2$, and $\mathcal{D}_3$. In general, problem P1 is hard to solve, because not only is it non-convex, but also the cells are highly coupled. From equation (\ref{sinr1}) in constraint (\ref{P1C1}), we can see that the inter-cell interference depends on the load levels, which in turn is governed by interference. What is more, the problem have different properties in different domains, and the impact of these domains should be discussed further.

\section{Optimization within a Cell} \label{Sec:SingleCell}

Let us first consider the simpler single-cell case, and later we will use the derived results to address the multi-cell problem. In the single-cell problem, we minimize the load of any generic cell ${i}$, whereas the load levels as well as the RIS reflection coefficients of the other cells are given. Let $\mathcal{L}_i$ $(\forall i \in \mathcal{I})$ represent the set of RISs in cell $i$. We define $\boldsymbol{\rho}_{-i} = \left\{{\rho}_{k}|k\in \mathcal{I},k\neq i\right\}$, $\boldsymbol{\phi}_{-i} = \left\{ {\phi}_{lm} |m\in \mathcal{M},  l\in \mathcal{L}\setminus\mathcal{L}_i\right\}$, and $\boldsymbol{\phi}_{i} = \left\{ {\phi}_{lm} |m\in \mathcal{M},  l\in \mathcal{L}_i\right\}$.

\subsection{Formulation and Transformation}
For a UE $j$ in cell $i$, the intended signal received at UE $j$ is given by
\begin{align}
&{ \Bigg | \underbrace{g_{ij}+ \sum_{l \in \mathcal{L}\setminus \mathcal{L}_i} \boldsymbol{G}_{il}  \boldsymbol{\Theta}_l \boldsymbol{H}_{lj}}_{\text{known}}+\underbrace{\sum_{l \in \mathcal{L}_i} \boldsymbol{G}_{il}  \boldsymbol{\Theta}_l \boldsymbol{H}_{lj}}_{\text{to be optimized}} \Bigg|^2 P_i}\notag \\
=& { \Bigg | \hat{g}_{ij}+{\sum_{l \in \mathcal{L}_i} \boldsymbol{G}_{il}  \boldsymbol{\Theta}_l \boldsymbol{H}_{lj}} \Bigg|^2 P_i},
\end{align}
where $\hat{g}_{i,j}$ is the known part of total gain.

Similarly, the interference received at UE $j$ is given by
\begin{equation}
 {\sum_{k\in\mathcal{I}, \atop k\neq i}\Bigg | \hat{g}_{kj}+{\sum_{l \in \mathcal{L}_i}\boldsymbol{G}_{kl}  \boldsymbol{\Theta}_l\boldsymbol{H}_{lj}}\Bigg|^2 P_k \rho_k}.
\end{equation}
Let $\boldsymbol{\Phi}_l = [\phi_{l1}, \phi_{l2}, ..., \phi_{lM}]^T$, we have 
\begin{equation}
 \boldsymbol{G}_{il}  \boldsymbol{\Theta}_l\boldsymbol{H}_{lj} = \boldsymbol{\Lambda}_{ijl}\boldsymbol{\Phi}_l,
\end{equation}
where $\boldsymbol{\Lambda}_{ijl} = \boldsymbol{G}_{il}\mathfrak{diag} \{\boldsymbol{H}_{lj}\} $. To distinguish from the SINR formulation (\ref{sinr1}) in the multi-cell problem, we use $\text{SINR}_j^{\left \langle s \right \rangle} \left(\boldsymbol{\phi}_{i} \right)$ as follows to present the SINR of UE $j$ in the single-cell problem.
\begin{align}\label{SSINR}
\text{SINR}_j^{\left \langle s \right \rangle} \left(\boldsymbol{\phi}_{i} \right)=  \frac{ \left | \hat{g}_{ij}+ \sum\limits_{l \in \mathcal{L}_i}\boldsymbol{\Lambda}_{ijl}\boldsymbol{\Phi}_l\right |^2{P_i} }{ \sum\limits_{k\in\mathcal{I}, \atop k\neq i} \left | \hat{g}_{kj}+ \sum\limits_{l \in \mathcal{L}_i}\boldsymbol{\Lambda}_{kjl}\boldsymbol{\Phi}_l\right |^2{P_k}\rho_k + \sigma^2}.
\end{align}
Then the single-cell problem is as follows.
\begin{subequations}\label{single_formulation}
\begin{align}
\textup{[\rm P2]}\ \ \ \underset{ \boldsymbol{\phi}_{i}}{\min} \ \ &\sum_{j\in \mathcal{J}_i} \frac{d_j}{\log_2\left(1+ \text{SINR}_j^{\left \langle s \right \rangle} \left( \boldsymbol{\phi}_{i} \right)  \right)} \label{P2obj}\\
\textup{ s.t.}\ \ \ & \phi_{lm} \in \mathcal{D}, \forall m \in \mathcal{M},\forall l \in \mathcal{L}_i. \label{P2C1}
\end{align}
\end{subequations}
Note that (\ref{P2obj}) is non-convex. To proceed, we first introduce auxiliary variables $\boldsymbol{\gamma}_i = [\gamma_1,\gamma_2,..., \gamma_{J_i}]^T$, where $J_i = |\mathcal{J}_i|$. Then problem P2 can be rewritten as follows.
\begin{subequations}
\begin{align}
\textup{[\rm P2.1]}\ \ \ \underset{ \boldsymbol{\phi}_{i}, \boldsymbol{\gamma}_i }{\min} \ \ &\sum_{j\in \mathcal{J}_i} \frac{d_j}{\log_2\left(1+ \gamma_j  \right)} \label{P2.1obj}\\
\textup{ s.t.}\ \ \ & \text{(\ref{P2C1})}, \notag\\
& \text{SINR}_j^{\left \langle s \right \rangle} \left( \boldsymbol{\phi}_{i} \right) \geq \gamma_j , \forall j \in \mathcal{J}_i.
 \label{P2.1C2}
\end{align}
\end{subequations}
Clearly, the objective function of P2.1 is convex. However, the additional constraint (\ref{P2.1C2}) is not convex. We introduce auxiliary variables $\boldsymbol{\beta}_i = [\beta_1,\beta_2,..., \beta_{J_i}]^T$ such that 
\begin{equation}\label{convexconstraint}
\sum_{k\in\mathcal{I}, k\neq i}\left | \hat{g}_{kj}+ \sum_{l \in \mathcal{L}_i}\boldsymbol{\Lambda}_{kjl}\boldsymbol{\Phi}_l\right |^2{P_k}\rho_k + \sigma^2\leq \beta_j,\forall j \in \mathcal{J}_i.
\end{equation}
\begin{lemma}
Constraint (\ref{convexconstraint}) is convex.
\end{lemma}
\begin{proof}
Note that 
\begin{align}
 &\left | \hat{g}_{kj}+ \sum_{l \in \mathcal{L}_i}\boldsymbol{\Lambda}_{kjl}\boldsymbol{\Phi}_l\right |^2 \notag\\
 =& \left(\hat{g}_{kj}+ \sum_{l \in \mathcal{L}_i}\boldsymbol{\Lambda}_{kjl}\boldsymbol{\Phi}_l\right)\left(\hat{g}_{kj}^\star+ \sum_{l \in \mathcal{L}_i} \boldsymbol{\Phi}^H_l\boldsymbol{\Lambda}_{kjl}^H\right)\notag\\
=&\sum_{l \in \mathcal{L}_i}\boldsymbol{\Phi}_l^H\boldsymbol{\Lambda}_{kjl}^H\boldsymbol{\Lambda}_{kjl}\boldsymbol{\Phi}_l  \!+\! 2\Re\left\{\hat{g}_{kj}^\star\sum_{l \in \mathcal{L}_i}\boldsymbol{\Lambda}_{kjl}\boldsymbol{\Phi}_l\right\} \!+\! | \hat{g}_{kj}|^2 ,
\end{align}
where $\sum_{l \in \mathcal{L}_i} \boldsymbol{\Phi}_l^H\boldsymbol{\Lambda}_{kjl}^H\boldsymbol{\Lambda}_{kjl}\boldsymbol{\Phi}_l$ is the second-order cone (SOC) and $2\Re\{\hat{g}_{kj}^\star\sum_{l \in \mathcal{L}_i}\boldsymbol{\Lambda}_{kjl}\boldsymbol{\Phi}_l\}$ is affine, hence constraint (\ref{convexconstraint}) is convex.
\end{proof}
With $\boldsymbol{\beta}_i$, the SINR constraint (\ref{P2.1C2}) can be expressed by
\begin{equation}\label{non-convexconstraint}
\left | \hat{g}_{ij}+ \sum_{l \in \mathcal{L}_i}\boldsymbol{\Lambda}_{ijl}\boldsymbol{\Phi}_l\right |^2{P_i} \geq \beta _j \gamma _j,\forall j \in \mathcal{J}_i.
\end{equation}
Thus, problem P2.1 can be restated as follows. 
\begin{subequations}\label{final_formulation}
\begin{align}
\textup{[\rm P2.2]}\ \ \ \underset{ \boldsymbol{\phi}_{i}, \boldsymbol{\gamma}_i, \boldsymbol{\beta}_i}{\min} \ \ &\sum_{j\in \mathcal{J}_i} \frac{d_j}{\log_2\left(1+ \gamma_j  \right)} \label{P2fobj}\\
\textup{ s.t.}\ \ \ & \text{(\ref{P2C1}), (\ref{convexconstraint}), and (\ref{non-convexconstraint})}. \notag
\end{align}
\end{subequations}
\begin{proposition}
Problems P2 and P2.2 are equivalent.
\end{proposition}
\begin{proof}
The result follows from that, by construction, there clearly is a one-to-one mapping between the solutions of P2 and P2.2.
\end{proof}

\subsection{Problem Approximation}

Consider P2.2, where constraints (\ref{P2C1}) and (\ref{non-convexconstraint}) are non-convex. We first deal with constraint (\ref{non-convexconstraint}), and defer the discussion of domain constraint (\ref{P2C1}).

We apply the transformation $\beta _j \gamma _j = \frac{1}{4}((\beta _j+\gamma_j)^2-(\beta _j-\gamma_j)^2)$ to constraint (\ref{non-convexconstraint}), we thus have the following equivalent constraint
\begin{align}\label{non-convexconstraint2}
\left(\beta _j+\gamma_j\right)^2 &-\left(\beta _j-\gamma_j\right)^2\notag\\
& - 4{P_i}\left | \hat{g}_{ij}+ \sum_{l \in \mathcal{L}_i}\boldsymbol{\Lambda}_{ijl}\boldsymbol{\Phi}_l\right |^2 \leq 0,\forall j \in \mathcal{J}_i.  
\end{align}
We define a function as follows
\begin{align}\label{definition1}
&\mathcal{F}_j(\gamma_j, \beta_j, \boldsymbol{\phi}_i) \triangleq \notag\\
& \left(\beta _j+\gamma_j\right)^2-\left(\beta _j-\gamma_j\right)^2 - 4{P_i}\left | \hat{g}_{ij}+ \sum_{l \in \mathcal{L}_i}\boldsymbol{\Lambda}_{ijl}\boldsymbol{\Phi}_l\right |^2, \forall j \in \mathcal{J}_i.
\end{align}
Observing that $\mathcal{F}_j$ is the difference of two convex (DC) functions, we adopt the DC programming technique:
\begin{enumerate}
\item The Taylor expansion of $\left(\beta _j-\gamma_j\right)^2$ at point $(\tilde{\beta}_j,\tilde{\gamma}_j)$ is given by
\begin{equation}\label{TE1}
	\left(\beta _j\!-\!\gamma_j\right)^2 \!=\! 2(\tilde{\beta} _j-\tilde{\gamma}_j)\left(\beta _j \!-\!\gamma_j\right) - (\tilde{\beta} _j-\tilde{\gamma}_j)^2 \!+\! R_1(\beta_j,\gamma_j),
\end{equation}
where $R_1(\beta_j,\gamma_j)$ is the remainder. As $\left(\beta _j-\gamma_j\right)^2$ is convex, we can easily know
\begin{equation}\label{TE2}
	R_1(\beta_j,\gamma_j) \geq 0.
\end{equation}
By (\ref{TE1}) and (\ref{TE2}), we have
\begin{equation}\label{equality1}
\left(\beta _j-\gamma_j\right)^2 \geq  2(\tilde{\beta} _j-\tilde{\gamma}_j)\left(\beta _j-\gamma_j\right) - (\tilde{\beta} _j-\tilde{\gamma}_j)^2,
\end{equation}
\item Similarly, we have
\begin{align}\label{equality2}
&\left | \hat{g}_{ij}+ \sum_{l \in \mathcal{L}_i}\boldsymbol{\Lambda}_{ijl}\boldsymbol{\Phi}_l\right |^2 \notag \\
=& \left(\hat{g}_{ij}+ \sum_{l \in \mathcal{L}_i}\boldsymbol{\Lambda}_{ijl}\boldsymbol{\Phi}_l\right)\left(\hat{g}_{ij}^\star+ \sum_{l \in \mathcal{L}_i} \boldsymbol{\Phi}^H_l\boldsymbol{\Lambda}_{ijl}^H\right)\notag\\
=& \left|\sum_{l \in \mathcal{L}_i}\boldsymbol{\Lambda}_{ijl}\boldsymbol{\Phi}_l \right|^2 + 2\Re\left\{\hat{g}_{ij}\sum_{l \in \mathcal{L}_i}\boldsymbol{\Phi}_l^H\boldsymbol{\Lambda}_{ijl}^H\right\}+| \hat{g}_{ij}|^2  \notag\\
\geq &2\Re\left\{ \left(\hat{g}_{ij}+\sum_{l \in \mathcal{L}_i} \boldsymbol{\Lambda}_{ijl}\tilde{\boldsymbol{\Phi}}_l\right)^\star \sum_{l \in \mathcal{L}_i}\boldsymbol{\Lambda}_{ijl}\boldsymbol{\Phi}_l\right\} \notag \\ 
& \qquad\qquad\qquad\qquad -\left|\sum_{l \in \mathcal{L}_i}\boldsymbol{\Lambda}_{ijl}\tilde{\boldsymbol{\Phi}}_l\right|^2+|\hat{g}_{ij}|^2.
\end{align}
\end{enumerate}
Then we define the following function
\begin{align}\label{definition2}
&\!\tilde{\mathcal{F}}_j(\gamma_j, \beta_j, \boldsymbol{\phi}_i) \!=\! \left(\beta _j\!+\!\gamma_j\right)^2\!-\!2(\tilde{\beta} _j\!-\!\tilde{\gamma}_j)\!\left(\beta _j\!-\!\gamma_j\right)\!+\! (\tilde{\beta} _j\!-\!\tilde{\gamma}_j)^2\notag\\
 &\qquad- 4{P_i}\Bigg( 2\Re\left\{ \left(\hat{g}_{ij}+\sum_{l \in \mathcal{L}_i} \boldsymbol{\Lambda}_{ijl}\tilde{\boldsymbol{\Phi}}_l\right)^\star \sum_{l \in \mathcal{L}_i}\boldsymbol{\Lambda}_{ijl}\boldsymbol{\Phi}_l\right\} \notag \\ 
 & \qquad\qquad\qquad\quad -\left|\sum_{l \in \mathcal{L}_i}\boldsymbol{\Lambda}_{ijl}\tilde{\boldsymbol{\Phi}}_l\right|^2+|\hat{g}_{ij}|^2\Bigg),\forall j \in \mathcal{J}_i.  
\end{align}
\begin{lemma}\label{ubf}
$\tilde{\mathcal{F}}_j(\gamma_j, \beta_j, \boldsymbol{\phi}_i) \geq {\mathcal{F}}_j(\gamma_j, \beta_j, \boldsymbol{\phi}_i)$.
\end{lemma}
\begin{proof}
It follows from immediately definitions (\ref{definition1}) and (\ref{definition2}), and the inequalities (\ref{equality1}) and (\ref{equality2}).
\end{proof}
Thus, we use the following constraint to approximate constraint (\ref{non-convexconstraint2}).
\begin{equation}\label{approximate constraints}
\tilde{\mathcal{F}}_j(\gamma_j, \beta_j, \boldsymbol{\phi}_i) \leq 0, \forall j \in \mathcal{J}_i.
\end{equation}
Noted that $\tilde{\mathcal{F}}_j(\gamma_j, \beta_j, \boldsymbol{\phi}_i)$ is a sum of a convex function and affine functions, thus the approximate constraint (\ref{approximate constraints}) is convex. We thus obtain an approximate problem as follows.
\begin{subequations}
\begin{align}
\textup{[\rm P2.3]}\ \ \ \underset{ \boldsymbol{\phi}_{i}, \boldsymbol{\gamma}_i ,\boldsymbol{\beta}_i }{\min} \ \ &\sum_{j\in \mathcal{J}_i} \frac{d_j}{\log_2\left(1+ \gamma_j  \right)} \label{objjjjj}  \\
\textup{ s.t.}\ \ \ &\text{(\ref{P2C1}), (\ref{convexconstraint}), and (\ref{approximate constraints})}. \notag
\end{align}
\end{subequations}
The objective function (\ref{objjjjj}), constraints (\ref{convexconstraint}) and (\ref{approximate constraints}) are all convex. Next, we investigate the effect of constraint (\ref{P2C1}).

\subsection{The Impact of Reflection Coefficient Models}

\subsubsection{Ideal}

As $\mathcal{D}_1$ is a convex set, approximate problem P2.3 is a convex problem for this domain. We can find the optimal solution efficiently via existing solvers. Making use of the following remark, the MM method ({a.k.a} Successive Upper-bound Minimization) can be applied to approach P2 and guarantee a local optimum \cite{wu1983convergence}.

\begin{remark}\label{remark}
By Lemma \ref{ubf}, the Lagrange dual function of problem P2.3 is an upper bound of that of P2.2.
\end{remark}

Specifically, we can find a locally optimal solution to P2 by solving a sequence of successive upper-bound approximate problems in form of P2.3. The algorithm for the single-cell case with $ \mathcal{D}_1$ is detailed in Algorithm \ref{al0}.

\begin{algorithm}[tbp]\label{al0}
\caption{Single-cell optimization based on MM} 
\KwIn{$\boldsymbol{\rho}_{-i}, \{d_j, g_{ij}, \boldsymbol{G}_{il}, \boldsymbol{H}_{lj}\}, \forall l \in \mathcal{L}_i, \forall j \in \mathcal{J}_i, \epsilon$;} 
\KwOut{$\rho_i$, $ \boldsymbol{\phi}_i$;}   
Initialize $\{ {\boldsymbol{\phi}}^{(0)}_i, {\boldsymbol{\gamma}}^{(0)}_i,{\boldsymbol{\beta}}^{(0)}_i\}$\;
$t\leftarrow0$\;
\Repeat
{{\rm (\ref{objjjjj}) or (\ref{objjjjjjj}) converges with respect to $\epsilon$}\label{step7}
}
{
$\{ \tilde{\boldsymbol{\phi}}_i, \tilde{\boldsymbol{\gamma}}_i,\tilde{\boldsymbol{\beta}}_i\}\leftarrow\{{\boldsymbol{\phi}}^{(t)}_i, {\boldsymbol{\gamma}}^{(t)}_i,{\boldsymbol{\beta}}^{(t)}_i\}$\;
{
Obtain $\{ {\boldsymbol{\phi}}^{(t+1)}_i, {\boldsymbol{\gamma}}^{(t+1)}_i,{\boldsymbol{\beta}}^{(t+1)}_i\}$ by P2.3 or P2.5\;
}
$t\leftarrow t+1$\;
}
$\rho_i = \sum_{j\in \mathcal{J}_i} \frac{d_j}{\log_2\left(1+ \gamma_j^{(t)}  \right)}$, $\boldsymbol{\phi}_i = \boldsymbol{\phi}_i^{(t)}$\;
\Return {$\rho_i $, $\boldsymbol{\phi}_i $}\;
\end{algorithm} 

\subsubsection{Continuous Phase Shifter}

Domain $ \mathcal{D}_2$ is non-convex, and we handle this issue by the penalty method. The resulting optimization problem can be rewritten as
\begin{subequations}
\begin{align}
\textup{[\rm P2.4]}\ \ \ \underset{ \boldsymbol{\phi}_{i}, \boldsymbol{\gamma}_i ,\boldsymbol{\beta}_i}{\min} \ \ &\sum_{j\in \mathcal{J}_i} \frac{d_j}{\log_2\left(1+ \gamma_j  \right)} - C \sum_{ l \in \mathcal{L}_i, \atop m\in \mathcal{M}}\left(\left|\phi_{lm}\right|^2 -1\right) 
 \label{P2.4obj}\\
\textup{ s.t.}\ \ \ &\text{(\ref{convexconstraint}) and (\ref{approximate constraints})}, \label{P2.4C1}\notag\\
& \left | \phi_m \right | \leq 1, \forall m \in \mathcal{M},\forall l \in \mathcal{L}_i,
\end{align}
\end{subequations}
where $C$ is the penalty parameter. Note that the penalty term $C \sum_{l \in \mathcal{L}_i,m\in \mathcal{M}}(\left|\phi_{lm}\right|^2 -1 )$ enforces that $\left|\phi_{lm}\right|^2 -1 = 0$ for the optimal solution of P2.4. However, the objective function becomes non-convex, and we use the DC programming technique to approximate it. With the first-order Taylor expansion of $\phi_{lm}^\star\phi_{lm}$, we have
\begin{equation}\label{1st}
\left|\phi_{lm}\right|^2 = \phi_{lm}^\star\phi_{lm} \geq 2\Re\left\{\tilde{\phi}_{lm}^\star   \phi_{lm} \right\}-|\tilde{\phi}_{lm} |^2.
\end{equation}
Then P2.4 can be approximated with the following convex problem. 
\begin{subequations}
\begin{align}
\textup{[\rm P2.5]}\ \underset{ \boldsymbol{\phi}_{i}, \boldsymbol{\gamma}_i ,\boldsymbol{\beta}_i}{\min} \ \ &\sum_{j\in \mathcal{J}_i} \frac{d_j}{\log_2\left(1+ \gamma_j  \right)}  - 2 C \sum_{  m \in \mathcal{M}, \atop l \in \mathcal{L}_i}  \Re\left\{\tilde{\phi}_{lm}^\star   \phi_{lm} \right\}\label{objjjjjjj}  \\
\textup{ s.t.}\ \ \ & \text{(\ref{convexconstraint}), (\ref{approximate constraints}), and (\ref{P2.4C1})} \notag.
\end{align}
\end{subequations}
Note that the constant term $C(|\tilde{\phi}_{lm} |^2+1)$ is not stated explicitly in (\ref{objjjjjjj}) as it has no impact on optimum.

We obtain a locally optimal solution to P2 with domain $\mathcal{D}_2$ by the MM method that solves a sequence of problems in form of P2.5. The algorithm is given in Algorithm \ref{al0}, as all major algorithmic steps for $\mathcal{D}_1$ remain.

\subsubsection{Discrete Phase Shifter}

If $\mathcal{D} = \mathcal{D}_3$, the single-cell problem P2 is NP-hard. 

\begin{proposition}\label{NP-hard}
	The single-cell problem P2 is NP-hard when $\mathcal{D} = \mathcal{D}_3$.
\end{proposition}
\begin{proof}
	Please refer to Appendix.
\end{proof}

As P2 with $\mathcal{D}_3$ is NP-hard, we aim to get an approximate solution to this case. We first obtain the solution to P2 with $\mathcal{D}_2$, $\phi_{lm}^{\left \langle \mathcal{D}_2 \right \rangle}$ $( \forall m \in \mathcal{M},\forall l \in \mathcal{L}_i)$, by Algorithm \ref{al0}. Then we make use of the following rounding equation to obtain an approximate solution $\phi_{lm}^{\left \langle \mathcal{D}_3 \right \rangle}$ $( \forall m \in \mathcal{M},\forall l \in \mathcal{L}_i)$ in this case.
\begin{equation}\label{rounding}
\phi_{lm}^{\left \langle \mathcal{D}_3 \right \rangle} = \underset{\phi_{lm}\in\mathcal{D}_3}{\text{argmin}} \left| \phi_{lm} -\phi_{lm}^{\left \langle \mathcal{D}_2 \right \rangle}\right|, \forall m \in \mathcal{M},\forall l \in \mathcal{L}_i,
\end{equation}
 
\subsection{Complexity Analysis}

Algorithm \ref{al0} includes two parts: 1) An iterative process based on the MM method, and 2) solving the convex approximate problems. We first analyze the computation complexity of solving problem P2.3 by a standard interior-point method in \cite{ben2001lectures}. There are $M \times |\mathcal{L}_i|$ SOC constraints of size two (in the number of variables), $|\mathcal{J}_i|$ SOC constraints of size $2\times M\times|\mathcal{L}_i|+1$, and $|\mathcal{J}_i|$ SOC constraints of size $M\times|\mathcal{L}_i| + 2$, with $2\times M\times|\mathcal{L}_i| + 2\times|\mathcal{J}_i|$ optimization variables. To simplify the notation, let $Q = M \times |\mathcal{L}_i| $ and $T = |\mathcal{J}_i|$. From \cite{6891348}, the complexity of solving P2.3 is thus given by 
\begin{align}
&\sqrt{Q+2T}\left(2Q+2T\right)\Big(4Q+T\left(2Q+1\right)^2 \notag\\
 & \qquad+T\left(Q+2\right)^2+\left(2Q+2T\right)^2\Big) = \mathcal{O}(Q^{3.5}T^{2.5}).
\end{align}
Similarly, the computation complexity of solving problem P2.5 is also of $\mathcal{O}(Q^{3.5}T^{2.5})$. Additionally, as the complexity of the iterative process based on the MM method is of $\mathcal{O}\left(ln(1/\epsilon)\right)$, where $\epsilon$ is the convergence tolerance in Step \ref{step7}, and the overall complexity of Algorithm \ref{al0} is of $\mathcal{O}\left(Q^{3.5}P^{2.5}\ln(1/\epsilon)\right)$


\section{Multi-cell Load Optimization} \label{Sec:multi-cell}

This section proposes an algorithmic framework for the multi-cell problem and then discuss its convergence. 

In the last section, we have solved the single-cell problem via algorithm \ref{al0}, and a locally optimal load of cell $i$ can be obtained, when loads of other cells $\boldsymbol{\rho}_{-i}$ and the reflection coefficients of RISs at the other cells $\boldsymbol{\phi}_{-i}$ are given. We aim to solve the multi-cell problem by solving iteratively the single-cell problem. However, Algorithm \ref{al0} will obtain different locally optimal solutions, with the different initial parameters. If we directly embed it into an algorithmic framework based on the fixed-point method like \cite{siomina2012analysis, 8353846}, the convergence will not be guaranteed. To avoid this issue, the initial parameters of Algorithm \ref{al0} should be pre-determined instead of random in the algorithmic framework. Therefore, we use $f_i\left(\boldsymbol{\rho}_{-i},\boldsymbol{\phi}_{-i}, \Psi_i \right) $ to represent the process of obtaining $\rho_i$ by the Algorithm \ref{al0} with the pre-determined initial parameters, {i.e.},
\begin{equation}
f_i\left(\boldsymbol{\rho}_{-i},\boldsymbol{\phi}_{-i}, \Psi_i \right) = \rho_i, \\
\end{equation}
where $\rho_i$ is obtained by Algorithm \ref{al0} with the initial parameters $\Psi_i =  \{ {\boldsymbol{\phi}}^{(0)}_i, {\boldsymbol{\gamma}}^{(0)}_i,{\boldsymbol{\beta}}^{(0)}_i\}$.

\begin{lemma}\label{lemma3}
$f_i\left(\boldsymbol{\rho}_{-i},\boldsymbol{\phi}_{-i}, \Psi_i \right)$ is well-defined.
\end{lemma}
\begin{proof}
Note that $\rho_i$ is obtained by solving a sequence of convex problems, where the solution to a problem is the initial point of the next problem instance. It is clear that Algorithm \ref{al0} will always converge by the construction of the MM method, and the same holds if it is applied repeatedly for a finite number of times. Therefore, $f_i\left(\boldsymbol{\rho}_{-i},\boldsymbol{\phi}_{-i}, \Psi_i \right)$ is well-defined.
\end{proof}

With $f_i\left(\boldsymbol{\rho}_{-i},\boldsymbol{\phi}_{-i}, \Psi_i \right)$ and Lemma \ref{lemma3}, we then propose an algorithmic framework based on the following iteration to obtain a locally optimal solution.
\begin{align}\label{iterationequality}
&\boldsymbol{\rho}^{(\tau+1)}= \Big[f_1\left(\boldsymbol{\rho}_{-1}^{(\tau)},\boldsymbol{\phi}_{-1}^{(\tau)}, \Psi_1^{(\tau)} \right),  f_2\left(\boldsymbol{\rho}_{-2}^{(\tau)},\boldsymbol{\phi}_{-2}^{(\tau)}, \Psi_2^{(\tau)} \right), \notag \\
& \qquad\qquad\qquad\qquad\qquad\quad ...,f_I\left(\boldsymbol{\rho}_{-I}^{(\tau)},\boldsymbol{\phi}_{-I}^{(\tau)}, \Psi_I^{(\tau)} \right)\Big]^T.
\end{align}
The algorithmic framework is detailed in Algorithm \ref{al1}, named iterative convex approximation (ICA) algorithm.

\begin{algorithm}[tbp]\label{al1}
\caption{Iterative convex approximation (ICA)} 
\KwIn{ $\{d_j, g_{ij}, \boldsymbol{G}_{il}, \boldsymbol{H}_{lj}\}, \forall i \in \mathcal{I}, \forall l \in \mathcal{L}, \forall j \!\in \mathcal{J}, \varepsilon$;} 
\KwOut{$\boldsymbol{\rho}$, $ \boldsymbol{\phi}$;}   
Initialize $\boldsymbol{\phi}^{(0)}$\label{step1}\;

Obtain $\boldsymbol{\rho}^{(0)}, \boldsymbol{\gamma}^{(0)} = \{ \boldsymbol{\gamma}_1^{(0)}, \boldsymbol{\gamma}_2^{(0)}, ..., \boldsymbol{\gamma}_I^{(0)}\}$, and $\boldsymbol{\beta}^{(0)} = \{ \boldsymbol{\beta}_1^{(0)}, \boldsymbol{\beta}_2^{(0)}, ..., \boldsymbol{\beta}_I^{(0)}\}$ from solving the multi-cell problem with $\boldsymbol{\phi}^{(0)}$\label{step2}\;
$\tau \leftarrow0$\;
\Repeat
{$||\boldsymbol{\rho}^{(\tau)}- \boldsymbol{\rho}^{(\tau-1)}||_{\infty} \leq \varepsilon$
}
{
\For{{\rm cell} $i$, $\forall i \in \mathcal{I}$}{${\rho}_i^{(\tau+1)} = f_i\left(\boldsymbol{\rho}_{-i}^{(\tau)},\boldsymbol{\phi}_{-i}^{(\tau)}, \Psi_i^{(\tau)} \right)$\;\label{s5}  $\boldsymbol{\phi}_{-i}^{(\tau+1)}\leftarrow \boldsymbol{\phi}_{-i}^{(\tau)}$ retrieved from Step \ref{s5}\; $\Psi_i^{(\tau+1)}\leftarrow \Psi_i^{(\tau)}$ retrieved from Step \ref{s5}\; }
$\tau \leftarrow \tau+1$\;
}
$\boldsymbol{\rho} = \boldsymbol{\rho}^{(t)}$, $\boldsymbol{\phi}=\boldsymbol{\phi}^{(t)}$\;
\Return {$\boldsymbol{\rho} $, $\boldsymbol{\phi}$}\;
\end{algorithm} 

Steps \ref{step1} and \ref{step2} are for obtaining appropriate initial parameters. We first give $\boldsymbol{\phi}^{(0)}$, and our problem be in form of the problem in \cite{siomina2012analysis}, then the $\boldsymbol{\rho}^{(0)}$ can be easily obtained, and the $\boldsymbol{\gamma}^{(0)}$ and $\boldsymbol{\beta}^{(0)}$ also can be calculated. The convergence of ICA is shown below.

\begin{lemma}\label{lemma4}
$f_i\left(\boldsymbol{\rho}_{-i},\boldsymbol{\phi}_{-i}, \Psi_i \right) $ is a monotonically increasing function of $\boldsymbol{\rho}_{-i}$.
\end{lemma}
\begin{proof}
When $\boldsymbol{\phi}_{-i}$ and $\Psi_i$ are fixed, our single-cell problem P2 will degenerate into the single-cell problem in \cite{siomina2012analysis}. In \cite{siomina2012analysis}, it has been proved that the corresponding function is monotonically increasing with increasing $\boldsymbol{\rho}_{-i}$.
\end{proof}

\begin{lemma}\label{lemma5}
$\rho_i^{(1)} \leq \rho_i^{(0)}, \forall i \in \mathcal{I}$ holds in ICA.
\end{lemma}
\begin{proof}
In ICA, $\boldsymbol{\rho}^{(1)}$ is obtained by (\ref{iterationequality}) based on the MM method with $\boldsymbol{\rho}^{(0)}$ as the initial parameter. By the construction of the MM method, the obtained locally optimal solution can not be worse than the initial one, i.e., $\rho_i^{(1)} \leq \rho_i^{(0)}, \forall i \in \mathcal{I}$.
\end{proof}

\begin{proposition}
ICA is convergent.
\end{proposition}
\begin{proof}
From Lemmas \ref{lemma4} and \ref{lemma5}, for any $\tau > 0$ and $\forall i \in \mathcal{I}$, ${\rho}^{(\tau+1)}_i \leq {\rho}^{(\tau)}_i $ holds. Therefore, the total load is monotonically decreasing over iterations in ICA, the convergence follows then from that the load levels are non-negative.
\end{proof}

A particular interesting aspect for the multi-cell problem is under what condition it can be solved to global optimum with our multi-cell algorithm. As an interesting finding, our algorithm carries the property of converging to global optimum, if the single-cell problem can be solved to optimum.

\begin{theorem}\label{additional}
	The global optimum of the multi-cell problem can be obtained if the single-cell problem can be solved to optimum.
\end{theorem}
\begin{proof}
First, we use $F_i(\boldsymbol{\rho}_{-i},\boldsymbol{\phi}_{i})$ to represent the optimal solution to the single-cell optimization problem under any given reflection coefficient vector $\boldsymbol{\phi}_i$, and it can be directly expressed by
\begin{align}
	F_i(\boldsymbol{\rho}_{-i},\boldsymbol{\phi}_{i}) = \sum_{j\in \mathcal{J}_i} \frac{d_j}{\log_2\left(1+ \text{SINR}_j^{\left \langle s \right \rangle} \left( \boldsymbol{\phi}_{i} \right)  \right)}.
\end{align}
By \cite[Theorem 2]{siomina2012analysis}, $F_i(\boldsymbol{\rho}_{-i},\boldsymbol{\phi}_{i})$ is strictly concave for $\boldsymbol{\rho}_{-i}$, and by \cite[Corollary 3]{siomina2012analysis}, we can prove easily it has the property of scalability, i.e., for any $\alpha>1$,
\begin{equation}\label{ine28}
	\alpha F_i(\boldsymbol{\rho}_{-i},\boldsymbol{\phi}_{i}) > F_i(\alpha \boldsymbol{\rho}_{-i},\boldsymbol{\phi}_{i}).
\end{equation}
In addition, $\text{SINR}_j^{\left \langle s \right \rangle} \left(\boldsymbol{\phi}_{i} \right)$ is clearly monotonically decreasing for $\boldsymbol{\rho}_{-i}$. Therefore, $F_i(\boldsymbol{\rho}_{-i},\boldsymbol{\phi}_{i})$ has {monotonicity}, i.e., for any $\boldsymbol{\rho}'_{-i} > \boldsymbol{\rho}_{-i}$,
\begin{equation}\label{ine39}
	F_i(\boldsymbol{\rho}'_{-i},\boldsymbol{\phi}_{i}) > F_i( \boldsymbol{\rho}_{-i},\boldsymbol{\phi}_{i}).
\end{equation}
These properties define the so called standard interference function (SIF) \cite{Yates}, hence, $F_i(\boldsymbol{\rho}_{-i},\boldsymbol{\phi}_{i})$ is an SIF.

Let $G_i(\boldsymbol{\rho}_{-i})$ be the optimal solution to the single-cell optimization problem P2, then we have
\begin{align}\label{fmin}
		G_i(\boldsymbol{\rho}_{-i}) &= F_i(\boldsymbol{\rho}_{-i},\boldsymbol{\phi}^*_{i}) = \underset{ \boldsymbol{\phi}_{i}}{\min} \ \rho_i \ \textup{s.t.}\  \text{(\ref{P2C1})}  \notag\\
		&= \min\left\{F_i(\boldsymbol{\rho}_{-i},\boldsymbol{\phi}_{i}) | \boldsymbol{\phi}_{i} \in \mathcal{D} \right\}
\end{align}
where $\boldsymbol{\phi}^*_{i}$ is the optimal reflection coefficients vector. We show $G_i(\boldsymbol{\rho}_{-i})$ is also an SIF, because it has scalability and monotonicity, and the proof is as follows.
\begin{itemize}
	 \item Scalability: \\
	 For any $\alpha>1$, let $\boldsymbol{\phi}^{\odot}_{i}$ be the optimal reflection coefficients vector for $F_i(\alpha\boldsymbol{\rho}_{-i},\boldsymbol{\phi}_{i})$, i.e., 
	 \begin{equation}\label{last1}
	 	G_i(\alpha \boldsymbol{\rho}_{-i}) = F_i(\alpha \boldsymbol{\rho}_{-i},\boldsymbol{\phi}^{\odot}_{i}) \leq F_i(\alpha \boldsymbol{\rho}_{-i},\boldsymbol{\phi}^{*}_{i}).
	 \end{equation}
	By (\ref{ine28}), 
	 \begin{equation}\label{last2}
	 	F_i(\alpha \boldsymbol{\rho}_{-i},\boldsymbol{\phi}^{*}_{i}) < \alpha F_i( \boldsymbol{\rho}_{-i},\boldsymbol{\phi}^{*}_{i}) = \alpha G_i( \boldsymbol{\rho}_{-i}).
	 \end{equation}
Therefore, by (\ref{last1}) and (\ref{last2}), we have
\begin{equation}
	G_i(\alpha \boldsymbol{\rho}_{-i}) < \alpha G_i( \boldsymbol{\rho}_{-i}).
\end{equation}

	\item Monotonicity: \\
	Since, for any $\boldsymbol{\phi}_{i} \in \mathcal{D}$ and $\boldsymbol{\rho}'_{-i} > \boldsymbol{\rho}_{-i}$, inequality (\ref{ine39}) holds, we have  
	\begin{equation}
		\min \{F_i(\boldsymbol{\rho}'_{-i},\boldsymbol{\phi}_{i}) | \boldsymbol{\phi}_{i} \in \mathcal{D} \} \! > \! \min\{F_i(\boldsymbol{\rho}_{-i},\boldsymbol{\phi}_{i}) | \boldsymbol{\phi}_{i} \in \mathcal{D} \}.
	\end{equation} 
	Therefore, for $\boldsymbol{\rho}'_{-i} > \boldsymbol{\rho}_{-i}$,
	\begin{equation}
		G_i(\boldsymbol{\rho}'_{-i}) > G_i(\boldsymbol{\rho}_{-i}),
	\end{equation}
\end{itemize}
Therefore, $G_i(\boldsymbol{\rho}_{-i})$ is an SIF, for which fixed-point iterations guarantees convergence and optimality \cite{Yates}. Hence the conclusion.
\end{proof}
  
 Theorem \ref{additional} implies that the multi-cell problem can be solved to global optimum if an optimal algorithm of single-cell is embedded into our algorithmic framework. The performance of the algorithmic framework remains satisfactory even though the embedded single-cell algorithm can only guarantee local optimality (e.g., Algorithm \ref{al0}). For small-scale scenarios with $\mathcal{D}_3$, the optimum of the single-cell problem can be found by exhaustive search, then the global optimum of the multi-cell problem also will be guaranteed by Theorem 10. With the global optimum provided, we will gauge the performance of the algorithmic framework embedded with Algorithm \ref{al0} in the next section.


\section{Performance Evaluation} \label{Sec:evaluation}

\subsection{Preliminaries}

We use a cellular network of seven cells, adopting a wrap-around technique. In each cell, ten UEs are randomly and uniformly distributed. The cells have the same number of RISs. All RISs have the same number of reflection elements. For each cell, its RISs are evenly distributed around the BS location, with a distance of $250$ m to the BS that is located in the center point of the cell. Each RIS has $M = 20$ reflection elements, and we use $S_i = | \mathcal{L}_i | \times M$ to denote the total number of reflection elements in cell $i$. The channel between cell $i$ and UE $j$ is given by $g_{ij} = D_{ij}^{-\alpha_{cu}} g_0$, where $D_{ij}$ is the distance between the BS and UE, $\alpha_{cu} = 3.5$ is the path loss exponent, and $g_0$ follows a Rayleigh distribution. Similarly, the channel from the BS of cell $i$ to RIS $l$ is given by $\boldsymbol{G}_{il} = D_{il}^{-\alpha_{ci}} \boldsymbol{G}_0$, and the channel from RIS $l$ to UE $j$ is given by $\boldsymbol{H}_{lj} = D_{lj}^{-\alpha_{iu}} \boldsymbol{H}_0$. Some additional simulation parameters are given in Table \ref{tab:test}.

\begin{table}[htbp]
 \caption{\label{tab:test}Simulation Parameters}
 \begin{center}
 \begin{tabular}{lll}
  \toprule
  \textbf{Parameter} & \textbf{Value}  \\
  \midrule
  Cell radius & $500$ m \\
  Carrier frequency & $2$ GHz \\
 Total bandwidth & $20$ MHz  \\
  RB power in each cell $P$ & $1$ W\\
  Noise power spectral density $\sigma ^2$ & $-174$ dBm/Hz  \\
 Convergence tolerance $\epsilon$ & $10^{-4}$ \\

  \bottomrule
 \end{tabular}
  \end{center}
\end{table}

\subsection{Benchmarks and Initialization}

In our simulation, we use three benchmark schemes as follows.
\subsubsection{No RIS} We take the system without RIS as a baseline scheme. For this scheme, we can compute the global optimum of resource minimization using the fixed-point method in \cite{siomina2012analysis}. 
\subsubsection{Random} In this scheme, the RIS reflection coefficients are randomly chosen from domain $\mathcal{D}_1$. Once the coefficients are set, the optimum subject to the chosen coefficient values can again be obtained by the fixed-point method.
\subsubsection{Decomposition} In this scheme, the RIS reflection coefficients of the cells are optimized independently, by treating inter-cell interference as zero or the worst-case value, namely Decomposition-1 and Decomposition-2, respectively. The resulting single-cell problem for domain $\mathcal{D}_1$ for a generic cell $i$ can be reformulated as follows.
\begin{subequations}
\begin{align}
\textup{[\rm P3]}\ \underset{ \boldsymbol{\phi}_{i}}{\min} \ \ &\sum_{j\in \mathcal{J}_i} \frac{d_j}{\log_2\left(1+ \frac{  | \hat{g}_{ij} + \sum_{l \in \mathcal{L}_i}\boldsymbol{\Lambda}_{ijl}\boldsymbol{\Phi}_l |^2{P_i} }{\Upsilon + \sigma^2}    \right)} \\
\textup{ s.t.}\ \ \ & \phi_{lm} \in \mathcal{D}_1, \forall m \in \mathcal{M},\forall l \in \mathcal{L}_i,
\end{align}
\end{subequations}
where given inter-cell interference $\Upsilon = 0$ or $\sum_{k\in\mathcal{I}, \atop k\neq i} \left | \hat{g}_{kj}+ \sum_{l \in \mathcal{L}_i}\boldsymbol{\Lambda}_{kjl}\boldsymbol{\Phi}_l\right |^2{P_k}$.
Following \cite{9039554}, by introducing a set of auxiliary variables $\boldsymbol{y}_i = [y_1, y_2, ..., y_{J_i} ]^T$, $\boldsymbol{a}_i = [a_1, a_2, ..., a_{J_i} ]^T$, and $\boldsymbol{b}_i = [b_1,b_2, ..., b_{J_i} ]^T$, problem P3 can be approximated by the following problem.
\begin{subequations}
\begin{align}
\textup{[\rm P3.1]}\underset{ \boldsymbol{\phi}_{i}, \boldsymbol{y}_i, \atop \boldsymbol{a}_i, \boldsymbol{b}_i}{\min} \ \ &\sum_{j\in \mathcal{J}_i} \frac{d_j}{\log_2\left(1+ \frac{  y_j{P_i} }{\sigma^2}    \right)} \label{P2.7obj}\\
\textup{ s.t.}\ \ \ & |\phi_{lm} |\leq 1, \forall m \in \mathcal{M},\forall l \in \mathcal{L}_i. \label{}\\
&a_j = \Re\left\{ \hat{g}_{ij}+ \sum_{l \in \mathcal{L}_i}\boldsymbol{\Lambda}_{ijl}\boldsymbol{\Phi}_l \right\}, \forall j\in \mathcal{J}_i\\
&b_j = \Im\left\{ \hat{g}_{ij}+ \sum_{l \in \mathcal{L}_i}\boldsymbol{\Lambda}_{ijl}\boldsymbol{\Phi}_l \right\}, \forall j\in \mathcal{J}_i\\
&y_j \leq \tilde{a}_j^2+\tilde{b}_j^2+2\tilde{a}_j({a}_j-\tilde{a}_j)+2\tilde{b}_j({b}_j-\tilde{b}_j), \notag \\ &\qquad\qquad\qquad\qquad\qquad\qquad\quad \forall j\in \mathcal{J}_i.
\end{align}
\end{subequations}

\begin{algorithm}[tbp]\label{al2}
\caption{Optimization for the decomposition scheme based on MM for cell $i$} 
\KwIn{ $\{d_j, g_{ij}, \boldsymbol{G}_{il}, \boldsymbol{H}_{lj}\}, \forall l \in \mathcal{L}_i, \forall j \!\in \mathcal{J}_i, \varepsilon$;} 
\KwOut{$\boldsymbol{\rho}$, $ \boldsymbol{\phi}$;}   
Initialize $ {\boldsymbol{\phi}}^{(0)}_i$\;
$t \leftarrow 0$\;
\Repeat
{{\rm (\ref{P2.7obj}) converges with respect to $\epsilon$}
}
{
$ \tilde{\boldsymbol{\phi}}_i \leftarrow {\boldsymbol{\phi}}^{(t)}_i$\;
$\tilde{a}_j \leftarrow \Re\left\{ \hat{g}_{ij}+ \sum_{l \in \mathcal{L}_i}\boldsymbol{\Lambda}_{ijl}\tilde{\boldsymbol{\Phi}}_l \right\}, \forall j\in \mathcal{J}_i$\;
$\tilde{b}_j \leftarrow \Im\left\{ \hat{g}_{ij}+ \sum_{l \in \mathcal{L}_i}\boldsymbol{\Lambda}_{ijl}\tilde{\boldsymbol{\Phi}}_l \right\}, \forall j\in \mathcal{J}_i$\;
Obtain ${\boldsymbol{\phi}}^{(t+1)}_i$ by solving P3.1\;
$t\leftarrow t+1$\;
}
\Return {$\boldsymbol{\phi}_i= \boldsymbol{\phi}_i^{(t)}$}\;
\end{algorithm}

P3.1 is a convex optimization problem that can be solved efficiently. An algorithm based on the MM method for problem P3 is provided in Algorithm \ref{al2}. After solving P3.1 for all individual cells to obtain the RIS reflection coefficients, the resulting optimum of P1 can be computed by the fixed-point method.

As for the value domain of RIS reflection coefficients, we consider $\mathcal{D}_1$, $ \mathcal{D}_2$, $\mathcal{D}_3$ (2-bit), and $\mathcal{D}_3$ (1-bit). Specifically, $\mathcal{D}_3$ (2-bit) represents the discrete phases with $N = 4$ and $\mathcal{D}_3$ (1-bit) that with $N = 2$. Additionally, in our simulation, we initialize all the reflection coefficients to be $\phi =  \mathrm{e}^{\mathrm{i} \pi}$. 

\subsection{Impact of Demand}

\begin{figure}[tbp]
\begin{center}
\begin{tikzpicture}
\begin{axis}[
    xlabel={Normalized demand $d$ (Mbps)},
    ylabel={Total load},
    xmin=0.4, xmax=0.8,
    ymin=0, ymax=6,
    xtick={0.4, 0.5, 0.6, 0.7, 0.8},
    ytick={0, 1, 2, 3, 4, 5, 6},
    legend pos=north west,
    grid style=densely dashed,
    tick label style={font=\footnotesize},
]

\addplot[ color=darkgray, mark=square, mark options={solid}, line width=0.8pt]
coordinates { (0.4, 2.00) (0.5, 2.70) (0.6, 3.47) (0.7, 4.32) (0.8, 5.25) };

\addplot[ color=darkgray, mark=o, densely dashed, mark options={solid}, line width=0.8pt]     
coordinates { (0.4, 1.94) (0.5, 2.62) (0.6, 3.40) (0.7, 4.20) (0.8, 5.14) };

\addplot[ color= teal, mark=x, densely dashed, mark options={solid}, line width=0.8pt]     
coordinates { (0.4, 1.11) (0.5, 1.8) (0.6, 2.60) (0.7, 3.50) (0.8, 4.5) };

\addplot[ color= teal, mark=square, mark options={solid}, line width=0.8pt]     
coordinates { (0.4, 1.85) (0.5, 2.25) (0.6, 2.78) (0.7, 3.30) (0.8, 3.78) };

\addplot[ color=red, mark=o, line width=0.8pt] 
coordinates { (0.4, 0.97) (0.5, 1.48) (0.6, 2.05) (0.7, 2.68) (0.8, 3.4) };

\addplot[ color=blue, mark=x, densely dashed, mark options={solid}, line width=0.8pt ]
coordinates { (0.4, 0.98) (0.5, 1.50) (0.6, 2.06) (0.7, 2.69) (0.8, 3.41) };

\addplot[ color=orange, mark=square, line width=0.8pt ]   
coordinates {(0.4, 1.03) (0.5, 1.54) (0.6, 2.12) (0.7, 2.77) (0.8, 3.52)};

\addplot[color=orange, mark=o, densely dashed, mark options={solid}, line width=0.8pt]
coordinates {(0.4, 1.17) (0.5, 1.73) (0.6, 2.36) (0.7, 3.07) (0.8, 3.87)};

\legend{No RIS, Random-$\mathcal{D}_1$, Decomposition-1-$\mathcal{D}_1$, Decomposition-2-$\mathcal{D}_1$, ICA-$ \mathcal{D}_1$, ICA-$\mathcal{D}_2$, ICA-$\mathcal{D}_3$ (2-bit), ICA-$\mathcal{D}_3$ (1-bit)}

\end{axis}
\end{tikzpicture}
\caption{Total load in function of normalized demand when $ S = 140$ and $ \alpha_{ci} = \alpha_{iu} = 2.2$.} \label{fig:1}
\end{center}
\end{figure}
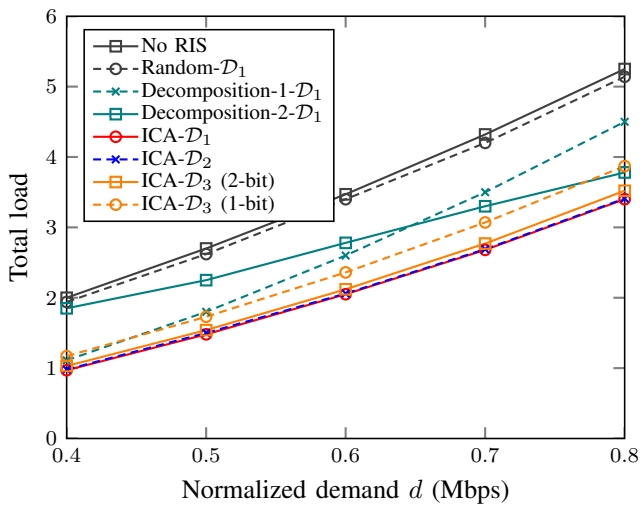

Fig. \ref{fig:1} illustrates the total load, i.e., resource consumption, with respect to the normalized demand $d$, when $ S_i = 140$ $(\forall i \in \mathcal{I})$ and $ \alpha_{ci} = \alpha_{iu} = 2.2$. It can be seen that the RIS provides little improvement if it is not optimized, as the difference between the performance of no RIS and that of the random scheme is very small. In contrast, ICA achieves significant performance gains compared with the three benchmark solutions. Specifically, ICA-$\mathcal{D}_1$ can save up to $51.3\%$ resource than no RIS when $d = 0.4$. Even with the most restrictive setup of ICA-$\mathcal{D}_3$ (1 bit), a $41.4\%$ reduction in resource usages is achieved. 

\subsection{Discussion for the coupling effects}

We have evaluated two decomposition schemes with given interference. In Fig. \ref{fig:1}, we can see that the performance of Decomposition-1 is close to that of ICA-$\mathcal{D}_1$ when $d \leq 0.4$, since small demand leads to fairly negligible inter-cell interference. As the demand increases, accounting for inter-cell interference becomes increasingly important, thus the difference between Decomposition-1 and ICA-$\mathcal{D}_1$ grows. For example, Decomposition-1 indicates $32.8\%$ more resource consumption compared with ICA-$\mathcal{D}_1$ when $d=0.8$. Conversely, the performance of Decomposition-2 is close to that of ICA-$\mathcal{D}_1$ when the demand is high, however, Decomposition-2 erroneously predicts $80\%$ more resource usage when $d = 0.4$. In summary, the results obtained by the two decomposition schemes both deviate significantly from those obtained when the coupling relation between cells is accounted for. The observation demonstrates the importance of capturing the dynamics due to load coupling in optimization.

\subsection{Impact of RIS}

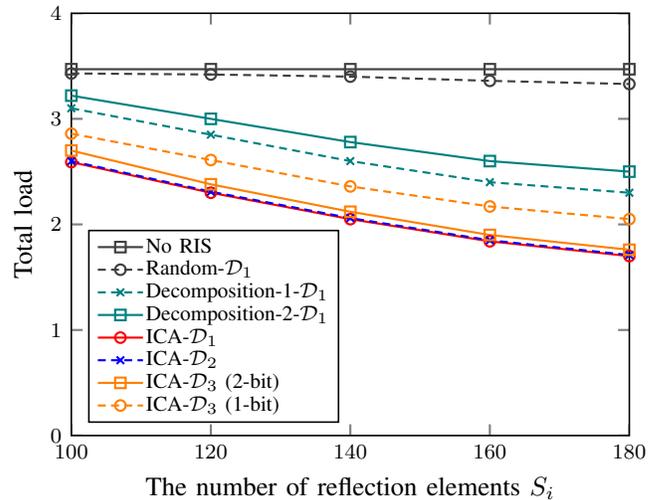
\begin{figure}[tbp]
\begin{center}
\begin{tikzpicture}
\begin{axis}[
    xlabel={The number of reflection elements $S_i$},
    ylabel={Total load},
    xmin=100, xmax=180,
    ymin=0, ymax=4,
    xtick={100, 120, 140, 160, 180},
    ytick={0, 1, 2, 3, 4},
    legend pos=south west,
    grid style=densely dashed,
    tick label style={font=\footnotesize},
]

\addplot[ color=darkgray, mark=square, mark options={solid}, line width=0.8pt]
coordinates { (100, 3.47) (120, 3.47) (140, 3.47) (160, 3.47) (180, 3.47) };

\addplot[ color=darkgray, mark=o, densely dashed, mark options={solid}, line width=0.8pt]     
coordinates { (100, 3.43) (120, 3.42) (140, 3.40) (160, 3.36) (180, 3.33) };

\addplot[ color= teal, mark=x, densely dashed, mark options={solid}, line width=0.8pt]     
coordinates { (100, 3.10) (120, 2.85) (140, 2.60) (160, 2.40) (180, 2.3) };

\addplot[ color= teal, mark=square, mark options={solid}, line width=0.8pt]     
coordinates { (100, 3.22) (120, 3.0) (140, 2.78) (160, 2.60) (180, 2.5) };

\addplot[ color=red, mark=o, line width=0.8pt] 
coordinates { (100, 2.59) (120, 2.30) (140, 2.05) (160, 1.84) (180, 1.7) };

\addplot[ color=blue, mark=x, densely dashed, mark options={solid}, line width=0.8pt ]
coordinates { (100, 2.60) (120, 2.31) (140, 2.06) (160, 1.85) (180, 1.71) };

\addplot[ color=orange, mark=square, line width=0.8pt ]   
coordinates {(100, 2.70) (120, 2.38) (140, 2.12) (160, 1.90) (180, 1.76)};

\addplot[color=orange, mark=o, densely dashed, mark options={solid}, line width=0.8pt]
coordinates {(100, 2.86) (120, 2.61) (140, 2.36) (160, 2.17) (180, 2.05)};

\legend{No RIS, Random-$\mathcal{D}_1$, Decomposition-1-$\mathcal{D}_1$, Decomposition-2-$\mathcal{D}_1$, ICA-$ \mathcal{D}_1$, ICA-$\mathcal{D}_2$, ICA-$\mathcal{D}_3$ (2-bit), ICA-$\mathcal{D}_3$ (1-bit)}

\end{axis}
\end{tikzpicture}
\caption{Total load in function of the number of reflection elements $ S_i$ $(\forall i \in \mathcal{I})$ when $d = 0.6$ Mbps, when $\alpha_{ci} = \alpha_{iu} = 2.2$.} \label{fig:2}
\end{center}
\end{figure}

Fig. \ref{fig:2} depicts the total load versus the number of reflection elements, when $d = 0.6$ Mbps and $ \alpha_{ci} = \alpha_{iu} = 2.2$. We can see that, with more reflected elements, the system can clearly benefit except for the random scheme. 

\begin{figure}[tbp]	
\begin{center}
\begin{tikzpicture}
\begin{axis}[
    xlabel={$\alpha_{RIS}$},
    ylabel={Total load},
    xmin=2, xmax=2.6,
    ymin=0, ymax=4,
    xtick={2, 2.2, 2.4, 2.6},
    ytick={0, 1, 2, 3, 4, 5, 6},
    legend pos=south east,
    grid style=densely dashed,
    tick label style={font=\footnotesize},
]

\addplot[ color=darkgray, mark=square, mark options={solid}, line width=0.8pt]
coordinates { (2, 3.47) (2.2, 3.47) (2.4, 3.47) (2.6, 3.47) };

\addplot[ color=darkgray, mark=o, densely dashed, mark options={solid}, line width=0.8pt]     
coordinates {  (2, 3.25) (2.2, 3.40) (2.4, 3.40) (2.6, 3.42) };

\addplot[ color= teal, mark=x, densely dashed, mark options={solid}, line width=0.8pt]     
coordinates {  (2, 2.15) (2.2, 2.60) (2.4, 2.85) (2.6, 3.00) };

\addplot[ color= teal, mark=square, mark options={solid}, line width=0.8pt]     
coordinates {  (2, 2.35) (2.2, 2.78) (2.4, 3.02) (2.6, 3.15) };

\addplot[ color=red, mark=o, line width=0.8pt] 
coordinates {  (2, 1.38) (2.2, 2.03) (2.4, 2.30) (2.6, 2.43) };

\addplot[ color=blue, mark=x, densely dashed, mark options={solid}, line width=0.8pt ]
coordinates {  (2, 1.39) (2.2, 2.04) (2.4, 2.32) (2.6, 2.455) };

\addplot[ color=orange, mark=square, line width=0.8pt ]   
coordinates { (2, 1.50) (2.2, 2.12) (2.4, 2.38) (2.6, 2.5)};

\addplot[color=orange, mark=o, densely dashed, mark options={solid}, line width=0.8pt]
coordinates { (2, 1.86) (2.2, 2.36) (2.4, 2.58) (2.6, 2.74)};

\legend{No RIS, Random-$\mathcal{D}_1$, Decomposition-1-$\mathcal{D}_1$, Decomposition-2-$\mathcal{D}_1$, ICA-$ \mathcal{D}_1$, ICA-$\mathcal{D}_2$, ICA-$\mathcal{D}_3$ (2-bit), ICA-$\mathcal{D}_3$ (1-bit)}

\end{axis}
\end{tikzpicture}
\caption{Total load in function of the path loss exponent $\alpha_{RIS} = \alpha_{ci} = \alpha_{iu}$, when $d = 0.6$ and $ S_i = 140$ $(\forall i \in \mathcal{I})$.} \label{fig:3}
\end{center}
\end{figure}
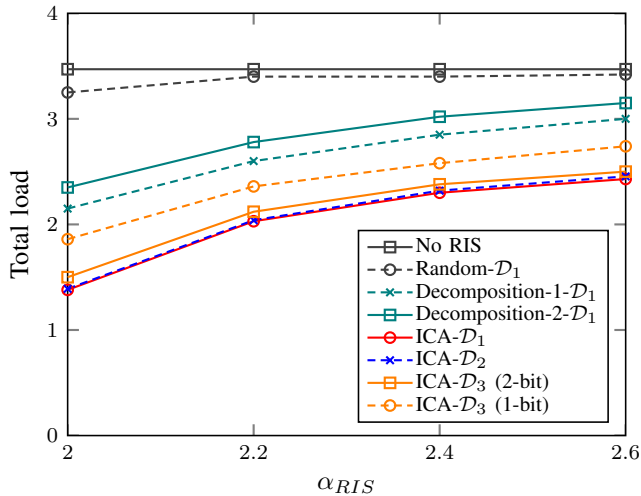

Fig. \ref{fig:3} shows the impact of the path loss exponent of RIS for the total load. As we can observe that the total load increases upon increasing $\alpha_{RIS}$, since larger $\alpha_{RIS}$ leads to the lower channel gain between RIS and BS/UE. The path loss exponent of RIS depends on the physical environment. For example, in practice $\alpha_{RIS}$ is usually smaller when the RIS is at a higher altitude due to fewer obstacles. But with an increasing height of the RIS position, the distance between the RIS and BS/UE also will increase, which leads to a larger path loss. 

We also observe from Fig. \ref{fig:1}-\ref{fig:3} that the curves of ICA-$\mathcal{D}_1$ and ICA-$\mathcal{D}_2$ almost overlap. It means that ICA loses little performance even though the amplitudes of the reflection coefficients are restricted to be one. In addition, for domain $\mathcal{D}_3$, we can see that, the performance of ICA-$\mathcal{D}_2$ is almost closed to the ICA-$\mathcal{D}_1$. It means that even RIS with low adjustability in practice can still help. 

\subsection{Discussion for Practical RIS}

\begin{figure}[tbp]
\begin{center}
\begin{tikzpicture}
\centering
\begin{axis}[
        ybar, axis on top,
        tick label style={font=\footnotesize},
        tick align=inside,
        xlabel={Normalized demand $d$ (Mbps)},
        ylabel={Total load},
        ymin=0, ymax=2,
        xtick={1, 2, 3, 4, 5},
        enlarge x limits=true,
       	legend style={at={(0.98, 0.98)},
	    legend pos=north west,
	    anchor=north west, 
    	legend columns=-1},
	legend image code/.code={\draw [#1] (0cm,-0.1cm) rectangle (0.2cm,0.25cm); },
]
    \addplot  coordinates { (1, 0.256) (2, 0.5715) (3, 0.9215) (4, 1.299) (5, 1.70) };
    \addplot  coordinates { (1, 0.246) (2, 0.56) (3, 0.9115) (4, 1.279) (5, 1.68) };
	\legend{ICA-$\mathcal{D}_3$ (2-bit), Optimal-$\mathcal{D}_3$ (2-bit)}
\end{axis}
\end{tikzpicture}
\caption{A comparison about the solution obtained by our algorithm and the optimal solution in a small scale simulation scenario under different normalized demand.} \label{fig:4}
\end{center}
\end{figure}
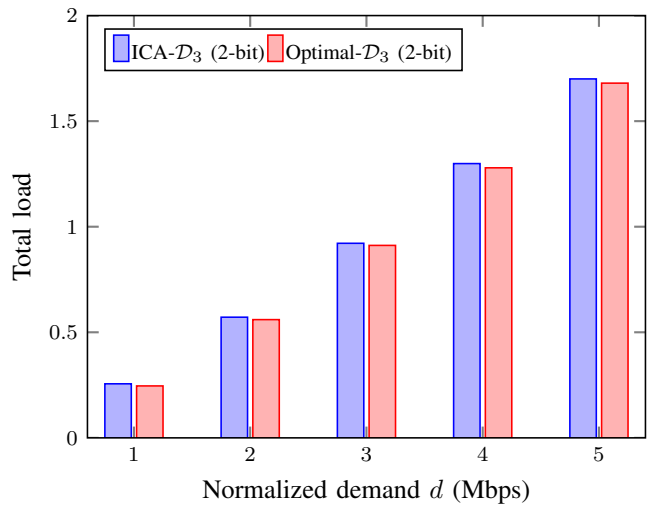

By Proposition \ref{additional}, for a small number of RIS elements, for the single-cell problem one can exhaustively consider the combinations of phase shifts for $\mathcal{D}_3$, and embedding this into the fixed-point method gives the global optimum for the multi-cell system. This allows us to gain insight of our algorithm in terms of performance with respect to global optimality. To this end, we consider a $3$-cell scenario where two UEs and total ten reflection elements in each cell. The results are shown in Fig. \ref{fig:4}, showing that algorithm ICA delivers solutions within at most 4\% gap to the global optimum for the small-scale scenario.

\subsection{Convergence Analysis}

\begin{figure}[tbp]
\begin{center}
\begin{tikzpicture}
\begin{semilogyaxis}[
    xlabel={Iteration $\tau$ in the ICA algorithm},
    ylabel={$||\boldsymbol{\rho}^{(\tau)}- \boldsymbol{\rho}^{(\tau-1)}||_{\infty}$},
    xmin=2, xmax=10,
    xtick={ 2, 3, 4, 5, 6, 7, 8, 9, 10},
    legend pos=north east,
    grid style=densely dashed,
    tick label style={font=\footnotesize},
]

\addplot[ color=darkgray, dashdotted, line width=0.8pt ] 
coordinates { (2, 1.5E-2) (10, 4E-9) };

\addplot[ color=red, dotted, line width=0.8pt ]
coordinates { (2, 3E-2) (10, 3E-10) };

\addplot[ color=blue, densely dashed, mark options={solid}, line width=0.8pt ]     
coordinates { (2, 2E-2) (10, 1E-10) };

\legend{  $d = 0.4\ S = 180$, $d = 0.4\ S = 100$, $d = 0.8\ S = 100$}

\end{semilogyaxis}
\end{tikzpicture}
\caption{This figure illustrates the norm $||\cdot||_{\infty}$ in function of iteration $\tau$ in the ICA algorithm under three scenarios.} \label{fig:5}
\end{center}
\end{figure}
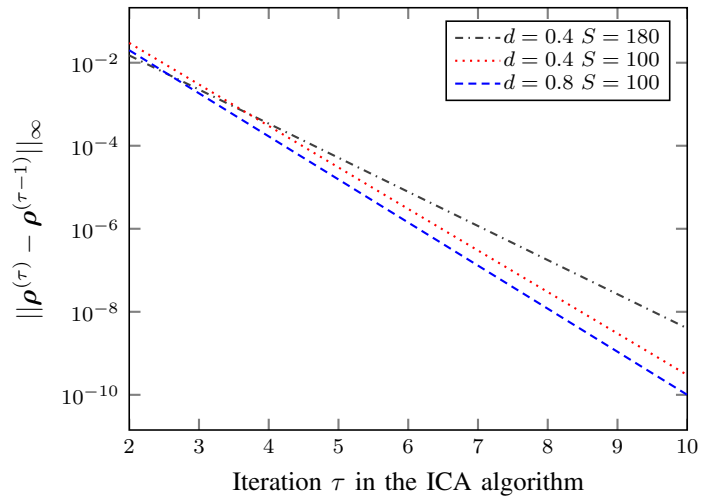

We show the convergence behavior of ICA in Fig. \ref{fig:5}. The convergence is consistently very fast. Note that the convergence becomes slightly slower for smaller demand. The reason is that smaller demand require a lower load in the final solution, and as a result, the difference between the initial one and the final load becomes larger. Similar reasoning applies to the convergence rate for larger $S$.


\section{Conclusion} \label{Sec:conclusion}

In this paper, we have investigated the resource minimization problem with interference coupling subject to the user demand requirement in multi-RIS-assisted multi-cell systems. An algorithmic framework has been designed to obtain a locally optimal solution. The numerical results have shown that the RIS can enhance the multi-cell system performance significantly. Additionally, even RIS with few discrete adjustable phases can achieve a good performance. In fact, RIS-assisted multi-cell scenarios are complex, and system simulation of large-scale scenario is time-consuming. Therefore, it would be desirable to explore model-based approaches, that are able to reasonably performance. In this paper, we chooese the load coupling model to capture the system dynamics, namely the interference dependency between interference and resource consumption. Although we acknowledge that the method is not exact, we hope one work represents the one step towards the multi-cell interference-coupled scenarios with RIS.


\section*{Appendix}
   
To prove Proposition \ref{NP-hard}, we use a polynomial-time reduction from the 3-satisfiability (3-SAT) problem to P2. The 3-SAT problem is NP-complete \cite{karp1972reducibility}. All notations we define in this appendix are used for this proof only.

Consider a 3-SAT instance with $n$ Boolean variables $x_1, x_2, ..., x_n$ and $m$ clauses. Denote by $\widehat{x}_i$ the negation of $x_i$. A literal means a variables or its negation. Each clause is composed by a disjunction of exactly three distinct literals, for example, $(\widehat{x}_1 \vee x_2 \vee \widehat{x}_n)$. The problem is to determine whether or not there exists an assignment of true/false values to the variables, such that all clauses are true. 

For any instance of 3-SAT with $n$ clauses and $m$ variables, we construct a corresponding network with a single cell, $n$ UEs, and $m$ RISs each having one reflection element, with domain $\mathcal{D}_3= \{  \mathrm{e}^{-\mathrm{i}\frac{\lambda}{4c}}, \mathrm{e}^{\mathrm{i}\frac{\lambda}{4c}}\}$, where $\lambda $ is the wavelength of the carrier, and $c$ is the speed of light. Each variable of the 3-SAT instance is mapped to an RIS. Moreover, the assignment of true/false to $x_i$ corresponds to setting the reflection coefficient of RIS $i$ to $\mathrm{e}^{\mathrm{i}\frac{\lambda}{4c}}$ and $\mathrm{e}^{-\mathrm{i}\frac{\lambda}{4c}}$, respectively. There is a one-to-one correspondence between the clauses in the 3-SAT instance and the UEs. It is not difficult to see that the parameters of our  P2 scenario can be set such that the following hold.

\begin{enumerate}
	\item The phase of the direct signal from the BS to all UEs is always at a crest. 
	\item For any UE, there are three RISs that potentially may contribute to the signal; the other RISs will have no effect on the UE (e.g., due to obstacle). In addition, For UE $j$, $j=1, 2, ... , m$, these three RISs are those corresponding to the three literals in clause $j$ of the 3-SAT instance. Specifically, the literal $x_i$ in clause $j$ corresponds the case where the phase of the reflected signal from RIS $i$ to UE $j$ will be at a crest if the reflection coefficient $ \mathrm{e}^{-\mathrm{i}\frac{\lambda}{4c}}$ is chosen, and that will be at a trough for $ \mathrm{e}^{\mathrm{i}\frac{\lambda}{4c}}$. Conversely, the literal $\widehat{x}_i$ in the clause corresponds the case where the phase will be at a trough if $ \mathrm{e}^{-\mathrm{i}\frac{\lambda}{4c}}$, and that will be at a crest if $ \mathrm{e}^{\mathrm{i}\frac{\lambda}{4c}}$. 
	\item The demand of UEs is set such that it can not be satisfied if and only if the power of the overall composite signal is zero. For any of three reflected signals received by UEs, the received power is one-third of the received power of the direct signal from the BS. By 1), we can know that, for any UE, the power of the overall composite signal will be zero (i.e., the demand can not be satisfied) if and only if its three reflected signals are all at a trough.
\end{enumerate}
Obviously, solving an instance of 3-SAT problem corresponds to determining if the defined P2 scenario has any feasible solution or not. Hence the conclusion.


%

\bibliographystyle{IEEEtran}
\bibliography{mybibtex}

\begin{IEEEbiographynophoto}{Zhanwei Yu} (Student Member, IEEE) received the B.Sc. in computer science and the M.Sc degrees in computer engineering from the Zhejiang University of Technology, Hangzhou, China, in 2016 and 2019, respectively. He is currently pursing the Ph.D degree with Department of Information Technology, Uppsala University, Sweden. His research interests include network optimization of B5G and 6G systems. 
\end{IEEEbiographynophoto}

\begin{IEEEbiographynophoto}{Di Yuan} (Senior Member, IEEE) received the M.Sc. degree in computer science and engineering and the Ph.D. degree in optimization from the Linköping Institute of Technology in 1996 and 2001, respectively. He has been an Associate Professor and then a Full Professor with the Department of Science and Technology, Linköping University, Sweden. In 2016, he joined Uppsala University, Sweden, as a Chair Professor. He has been a Guest Professor with the Technical University of Milan (Politecnico di Milano), Italy, since 2008, and a Senior Visiting Scientist with Ranplan Wireless Network Design Ltd., U.K., in 2009 and 2012. From 2011 and 2013, he was part time with Ericsson Research, Sweden. From 2014 and 2015, he was a Visiting Professor with the University of Maryland, College Park, MD, USA. He has been with the management committee of four European Cooperation in field of scientific and technical research (COST) actions, an Invited Lecturer of the European Network of Excellence EuroNF, and the Principal Investigator of several European FP7 and Horizon 2020 projects. His current research mainly addresses network optimization of 4G and 5G systems, and capacity optimization of wireless networks. He was a co-recipient of the IEEE ICC’12 Best Paper Award, and the Supervisor of the Best Student Journal Paper Award by the IEEE Sweden Joint VT-COM-IT Chapter in 2014. He is an Area Editor of {\em Computer Networks}.
	
\end{IEEEbiographynophoto}

\end{document}